\documentclass[pra,aps,superscriptaddress,twocolumn,nopacs,longbibliography,nofootinbib]{revtex4-1}
\usepackage{graphicx,color,epstopdf}
\usepackage{amsmath,amsfonts,enumerate,amsthm,amssymb,mathtools}
\usepackage{enumitem}
\usepackage{thmtools,thm-restate}
\usepackage{bbold}
\usepackage{hyperref}
\usepackage{multirow}
\usepackage{subfigure} 
\usepackage{mathdots}  
\usepackage{enumitem}  

\hypersetup{
   colorlinks=true,      
   linkcolor=blue,       
   citecolor=blue,       
   urlcolor=blue         
}

\def\A{ {\cal A} }

\def\E{ {\cal E} }
\def\L{ {\cal L} }
\def\H{ {\cal H} }
\def\L{ {\cal L} }

\def\R{ \mathbb{R} }

\def\>{\rangle}
\def\<{\langle}
\newcommand{\bra}[1]{\langle {#1} |}
\newcommand{\ket}[1]{| {#1} \rangle}
\newcommand{\abs}[1]{\left| {#1} \right|} 
\newcommand{\ketbra}[2]{\ensuremath{\left|#1\right\rangle\!\!\left\langle#2\right|}}
\newcommand{\matrixel}[3]{\ensuremath{\left\langle #1 \vphantom{#2#3} \right| #2 \left| #3 \vphantom{#1#2} \right\rangle}}
\newcommand{\tr}[2]{\mathrm{Tr}_{#1}\left( #2 \right)}
\newcommand{\iden}{\mathbb{1}}
\renewcommand{\v}[1]{\ensuremath{\boldsymbol #1}}

\newcommand{\be}{\begin{equation}}
\newcommand{\ee}{\end{equation}}

\definecolor{nicepurple}{rgb}{0.3,0,0.6}

\def\Re{\operatorname{Re}}
\def\Im{\operatorname{Im}}
\usepackage{pifont}
\newcommand{\om}[2]{(\omega_{#1 #2})}
\newcommand{\omp}[2]{(\omega_{#1'\! #2})}
\newcommand{\ompp}[2]{(\omega_{#1'\! #2 '\!})}
\def\P{\mathcal{P}}

\theoremstyle{plain}
\newtheorem{thm}{Theorem}
\newtheorem{lem}[thm]{Lemma}
\newtheorem{corol}[thm]{Corollary}
\theoremstyle{definition}

\theoremstyle{remark}

\newtheorem{ex}{Example}

\begin{document}

\title{Markovian evolution of quantum coherence under symmetric dynamics}
\author{Matteo Lostaglio}
\affiliation{Department of Physics, Imperial College London, London SW7 2AZ, United Kingdom}
\affiliation{ICFO-Institut de Ciencies Fotoniques, The Barcelona Institute of Science and Technology, Castelldefels (Barcelona), 08860, Spain}
\author{Kamil Korzekwa}
\affiliation{Department of Physics, Imperial College London, London SW7 2AZ, United Kingdom}
\affiliation{Centre for Engineered Quantum Systems, School of Physics, The University of Sydney, Sydney, NSW 2006, Australia}
\author{Antony Milne}
\affiliation{Department of Physics, Imperial College London, London SW7 2AZ, United Kingdom}
\affiliation{Department of Computing, Goldsmiths, University of London, New Cross, London SE14 6NW, United Kingdom}

\begin{abstract}

Both conservation laws and practical restrictions impose symmetry constraints on the dynamics of open quantum systems. In the case of time-translation symmetry, which arises naturally in many physically relevant scenarios, the quantum coherence between energy eigenstates becomes a valuable resource for quantum information processing. In this work we identify the minimum amount of decoherence compatible with this symmetry for a given population dynamics. This yields a generalisation to higher-dimensional systems of the relation $T_2 \leq 2 T_1$ for qubit decoherence and relaxation times. It also enables us to witness and assess the role of non-Markovianity as a resource for coherence preservation and transfer. Moreover, we discuss the relationship between ergodicity and the ability of Markovian dynamics to indefinitely sustain a superposition of different energy states. Finally, we establish a formal connection between the resource-theoretic and the master equation approaches to thermodynamics, with the former being a non-Markovian generalisation of the latter. Our work thus brings the abstract study of quantum coherence as a resource towards the realm of actual physical applications. 
\end{abstract}

\maketitle

\section{Introduction}
\label{sec:intro}
	
The consequences of symmetry in physics are of the utmost importance. Considerable insight into the evolution of a complex system can often be gained through an understanding of the underlying symmetries, even when the precise dynamics are not fully known or are too complicated to be solved exactly. One fundamental class of dynamics consists of those that are symmetric under time translations. This restriction originates from a conservation law for energy~\cite{ahmadi2013wigner,marvian2014extending}, the lack of a shared reference frame for time~\cite{bartlett2007reference}, or a superselection rule~\cite{gour2008resource}. Moreover, we will see that common assumptions made in the study of open quantum systems, such as the secular approximation~\cite{breuer2002open}, can be rephrased as symmetry constraints on the system dynamics. Such dynamics arise naturally in many areas of quantum physics, including the resource-theoretic formulation of thermodynamics~\cite{janzing2000thermodynamic, brandao2011resource, lostaglio2015description}, quantum metrology~\cite{bartlett2007reference, hall2012does, marvian2016quantify}, quantum noise of amplifiers~\cite{marvian2014modes}, cloning~\cite{dariano2001optimal, dariano2003optimal}, non-locality~\cite{verstraete2003quantum}, cryptography~\cite{verstraete2003quantum, bartlett2007reference} and quantum speed limits~\cite{marvian2016quantum}.
	
It is understood that symmetries of a system interacting with an environment can be studied within a framework that generalises Noether's theorem~\cite{marvian2014extending}. However, what are the general consequences of symmetries for open dynamics that are Markovian? The importance of this question is two-fold. Firstly, from a fundamental perspective, we wish to understand the interplay between memory effects and symmetries of the dynamics. Secondly, from the point of view of applications, it is crucial to unify the recent resource-theoretic approach~\cite{marvian2014extending, streltsov2016quantum} with the master equation formalism. This unification is particularly important for furthering research in fields such as quantum thermodynamics, in which the two approaches are currently very much distinct and rather disconnected.

In this work we focus on symmetry under time translations, a property characterising dynamical evolutions whose action is insensitive to their particular timing. Such dynamics constrain possible transformations of coherence, which then becomes an essential resource in quantum information processing \cite{gour2008resource, marvian2016quantify}. A central question is therefore: what is the minimal amount of decoherence compatible with a given population dynamics (e.g., relaxation process)? The main technical contribution of this paper is Theorem~\ref{thm:markovian_bound}, which gives the optimal coherence evolution consistent with a given population dynamics among all time-translation symmetric and Markovian processes. We also present several applications to illustrate the utility of our result.

Our study relies on a seminal work of Holevo~\cite{holevo1993note} that investigated the structure of covariant dynamical semigroups. In contrast to much of the literature that followed~\cite{holevo1995structure, vacchini2002quantum, vacchini2001translation, vacchini2005theory, hornberger2006master}, here we focus on finite-dimensional systems and our results on decoherence emerge directly from the underlying symmetry of the dynamics rather than the behaviour of a particular model. Moreover, our perspective on the problem is based on a resource-theoretic treatment of coherence, and thus we study the optimal limits of coherence processing. 

We begin in Sec. \ref{sec:elementaryexample} and \ref{sec:assumptions} by introducing more precisely the dynamics we will investigate and the underlying assumptions of Markovianity and time-translation symmetry. Sec.~\ref{sec:main_theorem} presents our main result, the minimal decoherence theorem, which forms the basis for the remainder of the paper. In Sec. \ref{sec:applications} we demonstrate the power of our result by applying it to several different scenarios. We first recover the famous relation $T_2 \leq 2 T_1$ for a qubit, and provide a generalisation of this inequality to $d$-dimensional systems. We then prove a relationship linking the complete loss of coherence to the existence of a unique stationary population, demonstrating that when the population does not relax to a unique fixed point there exist processes that indefinitely support coherence, despite non-trivial interaction with the environment and the absence of memory effects. This is followed by an investigation of non-Markovianity as a resource for coherence manipulation and an analysis of the role of non-Markovianity in the resource theory of thermal operations. We then show how observed coherences and populations can be used to witness non-Markovian behaviour in the evolution of a quantum system. Finally, we apply our result to relate the resource-theoretic formulation of quantum thermodynamics to the standard approach of open system dynamics, and to obtain tighter and physically more relevant bounds on the possible transformations under thermal operations. Overall conclusions are then given in Sec.~\ref{sec:conclusions}.

\section{Elementary example}
\label{sec:elementaryexample}

Before we give a formal statement of the systems studied in this paper, let us present an elementary example to give a flavour of our investigation. Consider a qubit system initially described by the density operator
\small
\begin{equation}
\label{eq:initialstate}
\rho(0)= \left[ \begin{array}{ccc}
	p(0) & c(0)  \\
	c^*(0) & 1-p(0)  
\end{array} \right],
\end{equation}
\normalsize
where $p(0)$ is the initial population of the ground state $\ket{0}$, and $c(0)$ is the initial quantum coherence between states $\ket{0}$ and $\ket{1}$. Now, assume that the system is in contact with an environment at thermal equilibrium. Under typical assumptions concerning the strength of interaction and environmental relaxation times (which will be discussed in more detail in Sec.~\ref{sec:physical}), the system evolves according to the Bloch equations~\cite{breuer2002open}
\small
\begin{equation}
\renewcommand*{\arraystretch}{2.5}
\label{eq:elementary}
\left\{
\begin{array}{lll}
\dfrac{dp}{dt} &=& L_{0|0} p(t) + L_{0|1} (1-p(t)),\\
\dfrac{d|c|}{dt} &=& - \gamma |c(t)|.
\end{array}
\right.
\end{equation}	
\normalsize
The populations transition rates satisfy $\sum_{x'} L_{x'|x} =0$ and $L_{x'|x}\geq 0$ for $x' \neq x$, whilst $\gamma\geq 0$ represents the decoherence rate.\footnote{Note that we have ignored the evolution of the phase of the off-diagonal term, $\arg c(t)$, since for the sake of our discussion we need only focus on the absolute value of the coherence term, $|c(t)|$.} 

Solving Eq.~\eqref{eq:elementary}, one finds
\small
\begin{equation}
\label{eq:elementary_sol}
\renewcommand*{\arraystretch}{1.5}
\left\{
\begin{array}{lll}
p(t)&=&\pi + (p(0) - \pi)e^{-t/T_1},\\
|c(t)|&=&e^{-t/T_2}|c(0)|,
\end{array}
\right.
\end{equation}
\normalsize
where $T_1 = 1/(|L_{0|0}| + L_{0|1})$ is the relaxation time, \mbox{$T_2 = 1/\gamma$} is the decoherence time and 
\begin{equation*}
\pi = \lim_{t\rightarrow \infty}p(t) =  L_{0|1}T_1
\end{equation*}
is the stationary ground state population. 

\begin{figure}
	\includegraphics[width=0.85\columnwidth]{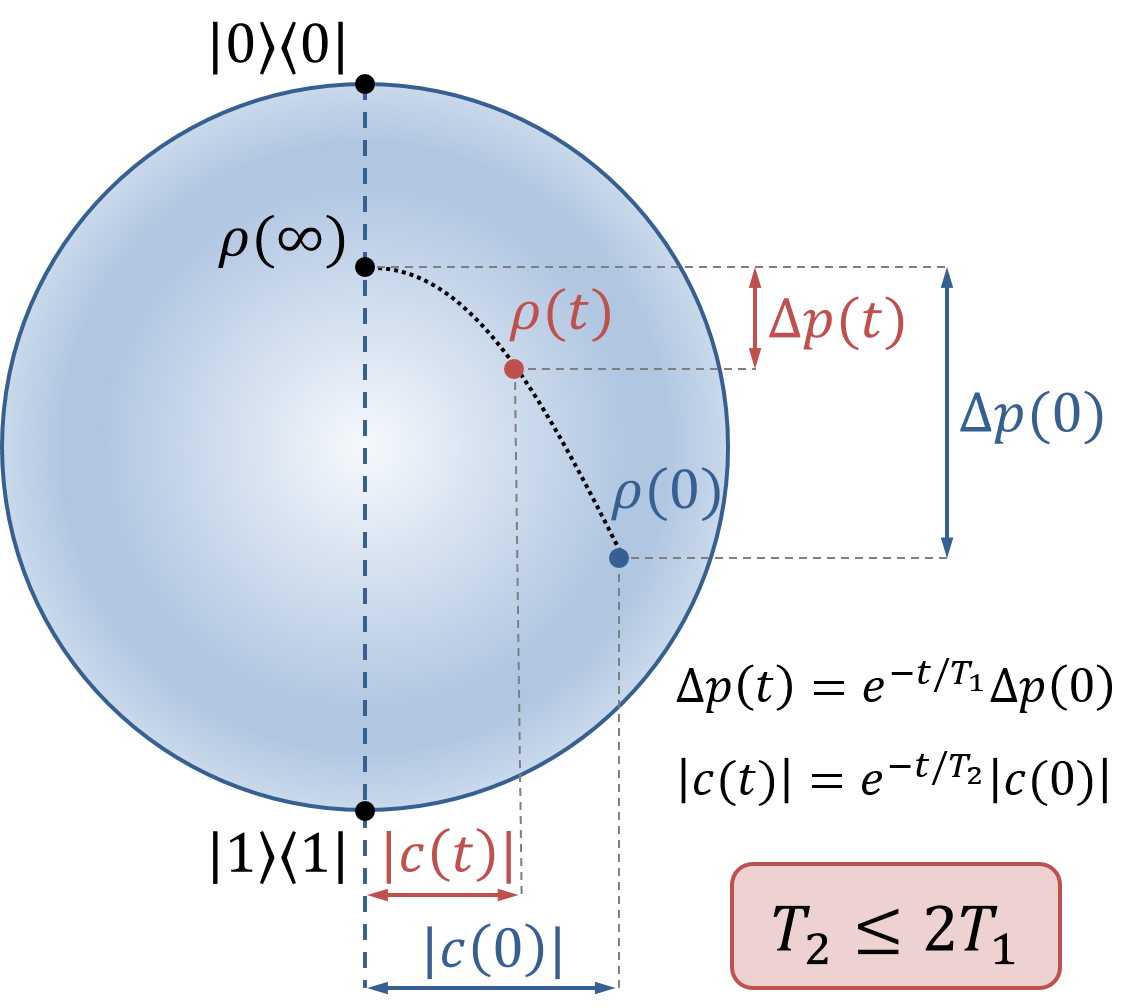}
	\caption{\label{fig:example} \emph{Elementary Example.} The evolution of the initial qubit state $\rho(0)$ towards the stationary state $\rho(\infty)$ presented on the Bloch sphere. During the evolution the difference between the current population and the stationary population, $\Delta p(t)=|p(t)-\pi|$, must decrease. Due to a constraint linking relaxation and decoherence processes, at any time the ratio between the current and initial coherence, $|c(t)|/|c(0)|$, is bounded by $\sqrt{\Delta p(t)/\Delta p(0)}$.}
\end{figure}

Crucially, the population relaxation process (described by $T_1$) and the decoherence process (described by $T_2$) are not independent. Loosely speaking, the reason for this is that every initial state must be mapped to a valid density matrix at all times, i.e., to a unit trace, positive semi-definite operator. More formally, one requires complete positivity of the map $\E$ that describes the evolution given in Eq.~{\eqref{eq:elementary}}, which in turn sets a general constraint linking the relaxation and decoherence times. In order to see this, note that complete positivity of $\E$ is equivalent to positivity of the Choi operator $J[\mathcal{E}]$~\cite{jamiolkowski1972linear,choi1975completely} (refer to Appendix~\ref{appendix:covariant} for details):
\small
\begin{equation*}
\renewcommand*{\arraystretch}{1.5}
J[\mathcal{E}] =\frac{1}{2}
\left[ \begin{array}{cccc}
	P_{0|0}(t) & e^{-t/T_2}  & 0 & 0  \\
	e^{-t/T_2} & P_{1|1}(t) & 0 & 0  \\
	0 & 0 & 1- P_{0|0}(t) & 0 \\
	0 & 0 & 0 & 1- P_{1|1}(t)
\end{array} \right],
\end{equation*}
\normalsize
where $P_{x'|x}(t)$ are the elements of population transition matrix \mbox{$P(t) = e^{L t}$}. Positivity of the Choi operator, $J[\mathcal{E}]\geq 0$, is thus equivalent to 
\begin{equation}
\label{eq:choi_positivity}
e^{-2t/T_2} \leq P_{0|0}(t) P_{1|1}(t)
\end{equation} 
at all times. A necessary condition for this is that the above inequality holds as $t \rightarrow 0$, which results in \mbox{$T_2 \leq 2 T_1$}. That this constraint is sufficient can be verified by substituting $T_2=2T_1$ into Eq.~\eqref{eq:choi_positivity} and checking directly that the inequality is satisfied at all times. This constraint, together with Eq.~\eqref{eq:elementary_sol}, links possible evolution of coherence with the evolution of population (see Fig.~\ref{fig:example}). 

In this elementary example it was possible to derive the relation between $T_1$ and $T_2$ by completely solving the dynamics. Of course, it becomes much harder to approach such a problem analytically beyond this simplest case of a qubit system. Moreover, $T_1$ and $T_2$ times are only properly defined when we deal with just two degrees of freedom. The aim of this work is to introduce a suitable theory that overcomes these limitations and allows us to study the links between population relaxation and decoherence processes for any finite-dimensional system. In so doing, we will see that the relation $T_2 \leq 2T_1$ is in fact a consequence of the time-translation symmetry of the dynamics.

\section{Setting the scene}
\label{sec:assumptions}

\subsection{Assumptions and the resulting structure}

Given a closed system described by a density operator $\rho$, its most general evolution is described by a unitary generated by some Hamiltonian $H$:
\begin{equation}
\label{eq:rho_closed_evol}
\frac{d\rho}{dt} = - i \H(\rho),
\end{equation}
where $\H(\rho):=[H,\rho]$. In many circumstances, however, the system is \emph{open}, i.e., it interacts with a generally large environment whose full quantum description is unmanageable. In this case one can still hope to describe the evolution of the system alone by means of a generalisation of Eq.~\eqref{eq:rho_closed_evol}. A standard way to do so is through the formalism of master equations~\cite{breuer2002open}. Within this work we assume that the evolution of quantum systems satisfies the following assumption: 

\begin{enumerate}[label=(A1)]
	\item \label{ass:markov}\emph{Markovianity}. The time evolution of the density operator $\rho$ is described by
	\begin{equation}
		\label{eq:rho_open_evol}
		\frac{d\rho}{dt} = - i \H(\rho) + \mathcal{L}(\rho),
	\end{equation}
	where, in addition to unitary evolution according to the Hamiltonian, we also have dissipative evolution governed by the \emph{Lindbladian} $\mathcal{L}$, whose general form was given in Refs.~\cite{lindblad1976generators,gorini1976completely},\footnote{Often in the literature $\L$ is called the \emph{dissipator} \cite{breuer2002open}, and $-i \H + \L$ is called the Lindbladian.}
	\begin{equation}
		\label{eq:lindbladian}
		\L(\cdot)=\A(\cdot)-\tfrac{1}{2}\{\A^{\dag}(\iden),\cdot\}.
	\end{equation}
	Here, $\A$ is a completely positive (CP) map, $\A^\dag$ denotes the adjoint of $\A$ with respect to the Hilbert-Schmidt inner product, \mbox{$\tr{}{\rho\A(\sigma)}=\tr{}{\A^{\dag}(\rho)\sigma}$}, and $\{\cdot,\cdot\}$ stands for the anticommutator. We denote the formal solution of Eq.~\eqref{eq:rho_open_evol} by
	\begin{equation}
	\label{eq:explind}
	\mathcal{E}_t(\rho):= e^{(-i\mathcal{H} + \mathcal{L})t}(\rho),	
	\end{equation} 
	with \mbox{$e^{-i\mathcal{H}t}(\rho)=e^{-i H t} \rho e^{iHt}$}. 
\end{enumerate}
	More technically, Eqs.~\eqref{eq:rho_open_evol}-\eqref{eq:lindbladian} describe the general evolution satisfying a semigroup property, which means that for any times $t_1$ and $t_2$ we have $\E_{t_1}(\E_{t_2}(\cdot))=\E_{t_1+t_2}(\cdot)$. Equivalently, the master equation has fixed and positive rates. There exists a wide range of axiomatic as well as microscopic derivations of this equation~(see Ref.~\cite{breuer2002open}). Here, we only remind the reader that~\ref{ass:markov} is linked to the fact that the environment is memoryless at the relevant timescales, which microscopically can be derived by assuming weak-coupling, sufficiently fast decay of environmental correlation functions and coarse-graining of time~\cite{breuer2002open}. Let us also note that evolutions satisfying~\ref{ass:markov} are sometimes referred to as  \emph{time-independent} or \emph{time-homogeneous} Markovian dynamics~\cite{rivas2014quantum}. 
	
The only other assumption that will be made in this work is that the dynamics are symmetric under time translations:
	\begin{enumerate}[label=(A2)]
	\item \label{ass:covariance}\emph{Time-translation symmetry.}  Each channel $\mathcal{E}_t$ is symmetric under time translations, which means that for every $s \in \mathbb{R}$ and $\rho$ we have
		\begin{equation}
			\label{eq:covariance}
			\mathcal{E}_t\left(e^{-i\mathcal{H}s}(\rho)\right) = e^{-i\mathcal{H}s} \left(\mathcal{E}_t(\rho)\right).
		\end{equation}
	\end{enumerate}
Since such channels are often called \emph{time-translation covariant}, for the sake of brevity we will sometimes simply refer to them as covariant channels. Another convention used in the literature is to call them \emph{phase-insensitive}~\cite{marvian2014modes}, as $e^{i\mathcal{H}s} \mathcal{E}_t e^{-i\mathcal{H}s} = \mathcal{E}_t$.

Note that, given \ref{ass:markov}, the assumption of time-translation symmetry \ref{ass:covariance} can be conveniently rewritten as a condition involving only the Lindbladian. Namely, for every $\rho$ we have
	\begin{equation}
	\label{eq:covariantwhenmarkov}
	\L(\H(\rho)) = \H(\L(\rho)).
	\end{equation}
Eq.~\eqref{eq:covariantwhenmarkov} lies at the core of how the symmetry properties of the dynamics translate into the symmetry of the corresponding generator $\L$. Let us make this more explicit. We identify the Hermitian operator $H$ as the Hamiltonian of a $d$-dimensional system (note that formally $H$ is the system Hamiltonian renormalised by the system-reservoir coupling~\cite{breuer2002open}). We further assume that $H$ is non-degenerate,
	\begin{equation}
	H=\sum_{x=0}^{d-1} \hbar \omega_x \ketbra{x}{x},
	\end{equation}
	and define the Bohr spectrum of $H$ as the set of all transition frequencies defined by the energy eigenvalues of $H$. In other words, it is the set $\{\omega\}$ such that there exist $\omega_x$ and $\omega_y$ in the spectrum of $H$ with $\omega_{xy}:=\omega_x - \omega_y = \omega$. Each $\omega$ denotes a particular \emph{mode}, consisting of matrix elements $\ketbra{x}{y}$ for which $\omega_{xy}=\omega$~\cite{marvian2014modes}. Now, the symmetry condition of Eq.~\eqref{eq:covariantwhenmarkov} enforces the Lindbladian to have a particular \emph{mode structure} dependent on $\{\omega\}$. As we explain in detail in Appendix~\ref{appendix:markovian} (using tools introduced in Appendix~\ref{appendix:covariant}), the matrix elements of a superoperator $\L$ vanish,
	\begin{equation}
	\bra{x'} \L(\ketbra{x}{y})\ket{y'}=0,
	\end{equation}
	unless $\omega_{xy}=\omega_{x'y'}$.  
	
	As a result, the evolution of populations (diagonal terms $\ketbra{x}{x}$ of a density matrix $\rho$) decouples from the evolution of coherences (off-diagonal terms $\ketbra{x}{y}$ of $\rho$), and the latter one can also be divided into independently evolving modes. To be more precise, let us first introduce the \emph{vector of populations} $\v{p}$ with components defined by \mbox{$p_x:=\rho_{xx}$}. Now, the evolution of $\v{p}(t)$ is fully described by the \emph{population transfer rate} matrix $L$,
	\begin{equation}
	\label{eq:motionforpopulation}
	\frac{d\v{p}}{dt}=L \v{p},
	\end{equation} 
	where the matrix elements of $L$ are given by
	\begin{equation}
		L_{x'|x} = \bra{x'} \L(\ketbra{x}{x})\ket{x'}.
	\end{equation}
	Note that since $\v{p}(t)=e^{Lt}\v{p}(0)$, the matrix $L$ is a generator of a stochastic matrix satisfying $L_{x|x}\leq 0$ and $\sum_{x'} L_{x'|x} = 0$ for all $x$~\cite{davies2010embeddable}.

\subsection{Physical significance of the symmetry condition}
\label{sec:physical}

The significance of the symmetry assumption~\ref{ass:covariance} may initially seem unclear since, despite its broad applicability, it is often referred to differently in different fields, and it is sometimes hidden within various physical approximations. Therefore, we will now briefly analyse the motivation behind it (see also Ref.~\cite{bartlett2007reference} and Sec.~IIIB of Ref.~\cite{marvian2016quantify}):
\begin{enumerate}
	\item Within quantum optics, time-translation symmetry is a consequence of the \emph{rotating-wave approximation} (RWA). This corresponds to manipulating the interaction Hamiltonian by discarding the so-called counter-rotating terms, which are those which rotate with higher frequency in the interaction picture. A typical example is the Jaynes-Cummings Hamiltonian, which in the simplest case reads \mbox{$H_{JC} \propto (\sigma_+ + \sigma_-)\otimes (a + a^\dag)$}, with $\sigma_{\pm}$ denoting qubit raising/lowering operators and $a,a^\dagger$ being bosonic annihilation and creation operators. In this case the approximation discards the terms $\sigma_-\otimes a$ and $\sigma_+ \otimes a^\dagger$, which leads to a master equation satisfying~\ref{ass:covariance}.
	\item Within the general theory of open quantum systems, \ref{ass:covariance} is called the \emph{secular approximation}. The secular approximation coincides with discarding terms in the Lindbladian that would prevent the commutation relation specified by Eq.~\eqref{eq:covariantwhenmarkov} from being satisfied. It is in fact a ``safer'' way of implementing the RWA \cite{fleming2010rotating} and is broadly used in applications~\cite{breuer2002open}. Another common instance in which Eq.~\eqref{eq:covariantwhenmarkov} holds is when one applies the RWA after Born-Markov approximation (see Sec.~3.3 of Ref.~\cite{breuer2002open}).

	\item In quantum metrology, consider the task of estimating the phase $\phi$ of a unitary generated by the Hamiltonian $H$, $U_\phi = e^{-i H \phi}$. Then assumption~\ref{ass:covariance} identifies the set of quantum channels $\{\mathcal{E}\}$ that degrade \emph{any} metrological resource $\rho$, i.e., for every $\rho$, optimal phase estimation using $\mathcal{E}(\rho)$ performs worse than optimal phase estimation using $\rho$~\cite{marvian2016quantify}.
	\item In quantum information, Eq.~\eqref{eq:covariance} coincides with the set of channels that can be performed in the absence of a reference frame for time \cite{bartlett2007reference}, or in the presence of a superselection rule for particle number. A dual perspective comes from the theory of $U(1)$-asymmetry, which is a resource theory where Eq.~\eqref{eq:covariance} defines the set of free operations~\cite{marvianthesis}. This is in fact a resource theory of quantum coherence in the basis defined by $H$ \cite{marvian2016quantify}. Time-translation covariance can also be linked to a global conservation law on energy \cite{keyl1999optimal} (see Theorem~25 of Ref.~\cite{marvianthesis}) and it is one of the defining properties of thermal operations \cite{lostaglio2015description}.
	\item In the study of quantum speed limits, the set of channels $\{\E\}$ covariant with respect to $H$ cannot increase the speed of evolution of any state under $H$. More precisely, the distinguishability between any state $\rho$ and its evolved version, $e^{-i\mathcal{H}t}(\rho)$, is lower-bounded by that between $\E(\rho)$ and $e^{-i\mathcal{H}t}(\E(\rho))$~\cite{marvian2016quantum}. 
\end{enumerate}

\section{Minimal decoherence theorem}
\label{sec:main_theorem}
The main result of this paper is to identify the optimal coherence preservation compatible with a given evolution of populations $\v{p}(t)$. The result is a sole consequence of the time-translation symmetry of the Lindbladian, as described by Eq.~\eqref{eq:covariantwhenmarkov}. More precisely, for a given  population transfer rate matrix $L$ we provide a bound that tells us what is the optimal amount of coherence that can be preserved in a state at time $t$. For notational convenience let us parametrise the matrix elements of a density matrix $\rho$ in the energy eigenbasis in the following way: \mbox{$\rho_{xy}=|\rho_{xy}|\vartheta_{xy}$}, where $\vartheta_{xy}$ is a phase factor, $|\vartheta_{xy}|=1$. We also define damping rates  \mbox{$\gamma_{x'y'}:=(|L_{x'|x'}| + |L_{y'|y'}|)/2$}, transport rates $t^{x'|x}_{y'|y}:= \sqrt{L_{x'|x} L_{y'|y}}$ and introduce the symbol $\sum^{(\omega)}_{x,y}$ to indicate the sum over indices of a mode $\omega$, i.e., $x,y$ such that $\omega_{xy}=\omega$. Then, we have: 

\begin{restatable}{thm}{markovianbound}
	\label{thm:markovian_bound}
	Let $\widetilde{\rho}_{x'y'}(t)$ be the solution of
	\begin{equation}
	\label{eq:boundrate}
	\frac{d\widetilde{\rho}_{x'y'}}{dt} = -\gamma_{x'y'} \widetilde{\rho}_{x'y'}  + \sum_{\mathclap{\substack{x\neq x'\\y \neq y'}}}^{\ompp{x}{y}} t^{x'|x}_{y'|y} \, \widetilde{\rho}_{xy},
	\end{equation}
	with $\widetilde{\rho}_{x'y'}(0) = |\rho_{x'y'}(0)|$. Then, if the time evolution of $\rho$ satisfies assumptions \ref{ass:markov}-\ref{ass:covariance} with population transfer rate matrix $L$, we have
	\begin{equation}
		\label{eq:bound}
	    |\rho_{x'y'}(t)| \leq \widetilde{\rho}_{x'y'}(t),
	\end{equation}
	for all $t \geq 0$.	Moreover, the bound can be saturated for all elements of a mode $\omega$ if for every $x',y',x,y$ with \mbox{$\omega_{x'y'} = \omega_{xy}=\omega$} one has 
				\begin{equation}
				\label{eq:phasematching}
				\vartheta_{x'y'}(0) \vartheta^*_{xy}(0) = \vartheta_{x'x}(0) \vartheta^*_{y'y}(0).
				\end{equation}

\end{restatable}

Eq.~\eqref{eq:phasematching} will be referred to as the \emph{Markovian phase-matching} condition for the initial state. We point out that pure states, and also mixed states for which amplitudes share a common phase (i.e., $\vartheta_{xy}=\vartheta$ for all $x$ and $y$), satisfy this condition for all modes. Moreover, the Markovian phase-matching condition is also satisfied independently of the initial state for modes consisting of a single element or two overlapping elements, i.e., $\rho_{xy}$ and $\rho_{x'y'}$ with $x=y'$ (see also Sec.~\ref{sec:coh_pres} for more details on overlapping elements). Finally, it is crucial to note that the evolution of $|\rho_{xy}|$ only depends on elements $|\rho_{x'y'}|$ with $\omega_{x'y'}=\omega_{xy}$ (recall that \mbox{$\omega_{xy}=\omega_x-\omega_y$}), which reflects the mode structure of the time-translation symmetric Lindbladian.

Physically, Theorem~\ref{thm:markovian_bound} demonstrates a combination of decay and transport phenomena, corresponding to each of the two terms in Eq.~\eqref{eq:boundrate} that contribute to the evolution of $\rho_{x'y'}$. The first is a \emph{decay term}, proportional to the amount of coherence $\rho_{x'y'}$ itself. If only this term were present then we would obtain a familiar exponential damping of coherence (with rate $\gamma_{x'y'}$), due to the presence of the dissipative environment. The extra contributions to the evolution of $\rho_{x'y'}$ are \emph{transport terms}. Only coherence elements $\rho_{xy}$ that rotate with the same frequency as $\rho_{x'y'}$ (i.e., belong to the same mode of coherence) can contribute, as indicated by the restricted summation. This ``selection rule'' is imposed by the underlying time-translation symmetry. The transport terms themselves have a suggestive physical interpretation. Namely, $L_{x'|x}$ is the transfer rate of the classical process that maps the energy state $x$ into $x'$, so $L_{x'|x}p_x(t)dt$ gives the population flow from $x$ to $x'$ between times $t$ and $t+dt$. The transfer of coherence from $\rho_{xy}$ to $\rho_{x'y'}$ involves a superposition of \emph{two} classical processes: the mapping of $x$ into $x'$ and of $y$ into $y'$. The optimal transport of coherence from $\rho_{xy}$ to $\rho_{x'y'}$ is characterised by the geometric mean of the transition rates of these two classical processes, i.e., $\sqrt{L_{x'|x} L_{y'|y}} \rho_{xy}(t) dt$.

The proof of Theorem~\ref{thm:markovian_bound} can be found in Appendix~\ref{appendix:proof} (however, refer also to Appendix~\ref{appendix:covariant}, where a step-by-step analysis of the structure of covariant maps can be found). In the next section we will proceed to present consequences and applications of the above theorem. However, let us first compare the bound on coherence processing under covariant Markovian dynamics, as specified by Theorem~\ref{thm:markovian_bound}, with the bound on general covariant dynamics, valid even when the Markovianity assumption is dropped. Similarly to the Markovian case, the evolution of populations under a general covariant map $\E$ is independent from the evolution of coherences. It may be described using the \emph{population transfer} matrix $P$ according to
\begin{equation}
\v{p}(t)=P\v{p}(0),
\end{equation}
where the matrix elements of $P$ are conditional probabilities given by 
\begin{equation}
P_{x'|x}=\bra{x'} \E(\ketbra{x}{x})\ket{x'}.
\end{equation}
Note that, in the case of a Markovian evolution, $P=e^{Lt}$. The bound we present below was first given in Ref.~\cite{lostaglio2015quantum}, however in Appendix~\ref{appendix:covariant} we provide a novel derivation that also sheds light on the tightness of the bound (see Appendix~\ref{app:proof_of_tightness} for details). 
\begin{restatable}{thm}{nonmarkovianbound}
	\label{thm:non_markovian_bound}
	Let $\sigma=\E(\rho)$, where $\E$ is a time-translation covariant CPTP map with corresponding population transfer matrix $P$. Then
	\begin{equation}
	\label{eq:non_markovian_bound}
	|\sigma_{x'y'}| \leq \sum_{x,y}^{\ompp{x}{y}} \sqrt{P_{x'|x} P_{y'|y}}|\rho_{xy}|.
	\end{equation}
	 The following tightness conditions hold:
	\begin{enumerate}
		\item Eq.~\eqref{eq:non_markovian_bound} can be simultaneously saturated for all $x', y'$ belonging to a given mode.
		\item  Eq.~\eqref{eq:non_markovian_bound} can be simultaneously saturated for all $x', y'$ if the Bohr spectrum is non-degenerate, or if the non-Markovian phase-matching condition holds, meaning that there exists a set of phase factors $\{\phi_x\}$ such that for all $x$ and $y$ we have $\vartheta_{xy}=\phi_x\phi_y^*$.
	\end{enumerate}
\end{restatable}
We emphasise that the non-Markovian phase matching condition is satisfied by all pure states and all mixed states $\rho$ with $\vartheta_{xy}=\vartheta$ for all $x$ and $y$.

\section{Applications}
\label{sec:applications}

\subsection{A generalisation of $T_1$ and $T_2$ times}
\label{sec:generalisation}

We begin by studying the evolution of a system with non-degenerate Bohr spectrum, i.e., described by a Hamiltonian for which all energy differences between any two levels are distinct. For such a system any off-diagonal element $\ketbra{x}{y}$ of the density matrix is the only element in its mode. Hence the evolution $\rho_{xy}(t)$ decouples from all other elements as
\begin{equation}
\frac{d |\rho_{xy}|}{dt} = - \Re(\alpha_{xy})|\rho_{xy}|,
\end{equation}
where $\alpha_{xy}$ can be found directly from the matrix of elements $\A$, the map that generates the Lindbladian [see Eq. \eqref{eq:lindbladian} and Eq.~\eqref{eq:drho_dt_exact_nondegenerate} in Appendix~\ref{appendix:proof}]. The decoherence rate $\Re(\alpha_{xy})$ enables us to define the decoherence time for $\rho_{xy}$ as $T_2^{xy}:=1/\Re(\alpha_{xy})$. The evolution of an off-diagonal element is thus given by
\begin{equation}
\label{eq:T2_evol}
|\rho_{xy}(t)|= |\rho_{xy}(0)|e^{-t/T_2^{xy}}.
\end{equation}

Now consider the evolution of diagonal elements. In terms of the population vector, we have $\v{p}(t)=e^{Lt}\v{p}(0)$, where $L$ is the population transfer rate matrix. Denote by $\lambda_x$ an eigenvalue of $L$, with corresponding eigenvector $\v v_x$ (i.e., we have $L\v v_x=\lambda_x \v v_x$ for $x=0,\dots, d-1$). Then, for diagonalisable $L$, the population vector evolves as
\begin{equation}
\label{eq:pop_evol1}
\v{p}(t) = \sum_{x=0}^{d-1} b_x e^{\lambda_x t} \v v_x,
\end{equation}
where $b_x$ are constants determined by the initial conditions. As $L$ is the generator of a stochastic matrix, it must have a zero eigenvalue, $\lambda_0=0$. Let us assume that this zero eigenvalue of $L$ is unique (nondegenerate), with eigenvector $\v{\pi}$. As we now show, this means that the population dynamics has a unique stationary distribution $\v{\pi}$, a situation sometimes referred to as \emph{ergodic} dynamics~\cite{roga2010davies}. For all non-zero eigenvalues we must have $\Re(\lambda_x)<0$ (following directly from the Gershgorin circle theorem~\cite{gershgorin1931uber,golub2012matrix}). Hence
\begin{equation}
\label{eq:pop_evol2}
\v{p}(t) = \v\pi+\sum_{x=1}^{d-1} b_x e^{-t/T_1^x} e^{i \Im(\lambda_x)t}\v v_x,
\end{equation}
where we have defined relaxation times \mbox{$T_1^x:=1/|\Re(\lambda_x)|$}. Clearly, as $t\rightarrow\infty$, we have $\v{p}(t)\rightarrow \v\pi$, so that the system relaxes towards a unique stationary population.\footnote{In fact, our analysis follows in much the same way for the case that $L$ is not diagonalisable. In this case, $L$ must have some eigenvalue $\lambda_x$ that is \textit{defective}, i.e., has multiplicity $m>1$ but possesses fewer than $m$ linearly independent eigenvectors. To form a complete solution to the differential equation for the evolution of populations we must then use \textit{generalised eigenvectors} $\v w_x$, and the solution, Eq.~\eqref{eq:pop_evol1}, will have terms of the form $q(t) e^{\lambda_x t} \v w_x$ where $q(t)$ is a polynomial function~\cite{edwards2008elementary}. Any defective eigenvalue $\lambda_x$ must be non-zero for ergodic $L$, and we also have $\Re(\lambda_x)<0$. Hence, in Eq.~\eqref{eq:pop_evol2} as $t\rightarrow\infty$, we still obtain relaxation to a fixed population $\v \pi$. Furthermore, we may still write $T_1^x:=1/|\Re(\lambda_x)|$ as a relaxation time, and Corollary~\ref{corol:generalisedT1T2} will still hold precisely as given.}

The following result gives a direct relation between the decoherence times $T_2^{xy}$ and the relaxation times $T_1^x$.
\begin{corol}
	\label{corol:generalisedT1T2}
Consider any system with non-degenerate Bohr spectrum evolving towards a unique stationary population. Then, under assumptions \ref{ass:markov}-\ref{ass:covariance}, we have the tight bound
\begin{equation}
\label{eq:generalisedT1T2}
	\langle T_2 \rangle_h \leq \frac{d}{d-1} \langle T_1 \rangle_h,
\end{equation}
where $\langle\cdot\rangle_h$ denotes the harmonic mean over all decoherence times $T_2^{xy}$ and all relaxation times $T_1^x$.
	\begin{proof}
Theorem~\ref{thm:markovian_bound} states that the evolution of an off-diagonal element is bounded as $|\rho_{xy}(t)|\leq\widetilde\rho_{xy}(t)$, where for a non-degenerate Bohr spectrum $\widetilde\rho_{xy}$ is the solution to $d\widetilde\rho_{xy}/dt=-\gamma_{xy}\widetilde\rho_{xy}$,
with $\widetilde{\rho}_{xy}(0) = |\rho_{xy}(0)|$. Hence, $\widetilde\rho_{xy}(t)=|\rho_{xy}(0)|e^{-\gamma_{xy}t}$. Comparing with Eq.~\eqref{eq:T2_evol} we see that $e^{-t/T_2^{xy}}\leq e^{-\gamma_{xy}t}$ and so $1/T_2^{xy}\geq\gamma_{xy}$. From Theorem~\ref{thm:markovian_bound} these inequalities are tight.

Consider now the trace of the population transfer rate matrix. We have $\tr{}{L}=\sum_{x=0}^{d-1}\lambda_x$. Since $L$ is a real matrix, its eigenvalues come in complex conjugate pairs. Hence, recalling that $T_1^x=1/|\Re(\lambda_x)|$ with $\Re(\lambda_x)<0$ and $x \neq 0$, we obtain 
\begin{equation}
\label{eq:trL_1}
|\tr{}{L}|=\sum_{x=1}^{d-1}\frac{1}{T_1^x} \Rightarrow \langle T_1 \rangle_h = \frac{d-1}{|\tr{}{L}|}.
\end{equation}
Notice that we also have $|\tr{}{L}|=\sum_{x=0}^{d-1}|L_{x|x}|$. Since $\gamma_{xy}=\frac12(|L_{x|x}|+|L_{y|y}|)$ and  $\gamma_{xy}\leq 1/T_2^{xy}$, we obtain
\begin{equation}
\nonumber
|\tr{}{L}|=\frac{2}{d-1}\sum_{x>y}\gamma_{xy} \leq \frac{2}{d-1}\sum_{x>y}\frac{1}{T_2^{xy}}=\frac{d}{\langle T_2 \rangle_h},
\end{equation}
where $\langle T_2 \rangle_h$ is the harmonic mean of the $\frac12d(d-1)$ decoherence times. Finally, we substitute this inequality into Eq.~\eqref{eq:trL_1} to give the result stated. 
	\end{proof}
\end{corol}

The tightness of Corollary~\ref{corol:generalisedT1T2} relies on saturating the bounds of Theorem~\ref{thm:markovian_bound}. In fact, in Appendix~\ref{appendix:proof} we explicitly show how one can construct $\A$ that leads to the longest possible decoherence time $T_2^{xy}=1/\gamma_{xy}$ for every coherence element; the resulting evolution achieves \mbox{$\langle T_2 \rangle_h = \frac{d}{d-1} \langle T_1 \rangle_h$.}

When we take the simplest case of a qubit, $d=2$, there is only one relaxation time and one decoherence time. Hence in Corollary \ref{corol:generalisedT1T2} there is no need to perform an average, and we obtain the well-known result \mbox{$T_2\leq2 T_1$}. For a qutrit, $d=3$, we instead obtain \mbox{$\langle T_2 \rangle_h \leq \frac32 \langle T_1 \rangle_h$}. Note that for large $d$, Corollary \ref{corol:generalisedT1T2} simply bounds the harmonic mean of decoherence times by the harmonic mean of relaxation times.

As an example application of our bound, we consider the case of thermalisation. When the population dynamics is ergodic and transfer rates satisfy the detailed balance condition (\mbox{$L_{x'|x}=L_{x|x'}e^{-\beta \hbar \omega_{x'x}}$} with \mbox{$\beta:=1/(k_B T)$} denoting the inverse temperature), the populations relax towards a thermal state, i.e., components of the stationary population are given by $\pi_x\propto e^{-\beta \hbar\omega_x}$. Since the columns of $L$ sum to zero, we may sum over all non-diagonal elements to obtain
\begin{equation*}
|\tr{}{L}|=\sum_{\mathclap{\substack{x,x'\\x \neq x'}}} L_{x'|x} = \sum_{\mathclap{\substack{x,x'\\x'<x}}} L_{x'|x} (1+e^{-\beta \hbar \omega_{xx'}}),
\end{equation*}
where in the second step we split the summation into elements with $x'<x$ and those with $x'>x$ and used the detailed balance condition. Eq.~\eqref{eq:trL_1} then gives directly an expression for $\langle T_1 \rangle_h$ and hence, according to Corollary~\ref{corol:generalisedT1T2}, a bound for the harmonic average of decoherence times. 

\subsection{Coherence preservation}

\label{sec:coh_pres}

The results of Ref.~\cite{lostaglio2015quantum}, strengthened in Theorem~\ref{thm:non_markovian_bound}, together with the tools developed here, allow us to assess the role of non-Markovianity in the preservation of coherence in the presence of a dissipative environment. Intuitively, one may expect that non-trivial Markovian processing of coherence necessarily yields deterioration of the quantum resources at hand, whereas non-Markovianity could provide an advantage. It is important to note that under assumption \ref{ass:covariance} \emph{alone}, coherence cannot be created in the system. Introducing
\begin{equation}
\label{eq:coherenceonenorm}
S_\omega (\rho) = \sum_{x,y}^{(\omega)}|\rho_{xy}|,
\end{equation}
one can show that for any quantum channel $\mathcal{E}$ satisfying \ref{ass:covariance}, one has \mbox{$S_\omega (\E(\rho)) \leq S_\omega (\rho)$} for every mode $\omega$~\cite{marvian2014modes}. 

Significantly, in Ref.~\cite{lostaglio2015quantum} it was shown that using covariant operations non-trivial processing of coherence (e.g., coherence transfer within a mode) can be performed perfectly, so that \mbox{$S_\omega (\E(\rho))=S_\omega (\rho)$}. On the other hand, typically considered noise models are Markovian~\cite{breuer2002open}. Hence, a more relevant question is: are there non-trivial covariant channels admitting a master equation description that preserve coherence indefinitely? The existence of such \emph{frozen coherence} has recently been proposed in Ref.~\cite{bromley2015frozen}, and an experimental demonstration followed shortly thereafter~\cite{silva2016observation}. Here, we formalise the question by asking what general features in the class of master equations satisfying assumption \ref{ass:covariance} allow such phenomena to arise. Our approach differs from that of Ref.~\cite{bromley2015frozen} in that here noise acts in the same basis in which the quantum information is encoded, rather than in a transversal basis. We also note that our investigation concerns quantum coherence between different energy eigenspaces, so we exclude the obvious possibility that superpositions can be preserved within decoherence-free subspaces.

Another way to phrase the question above is: do perfect (covariant) coherence manipulations necessarily require non-Markovianity? We begin to answer this question with the following result:

\begin{restatable}{corol}{attractor}
	\label{corol:attractor2}
	Under assumptions  \ref{ass:markov}-\ref{ass:covariance}, assume the population dynamics has a unique stationary distribution $\v{\pi}$ with $\pi_x \neq 0 \;\forall x$. Then 
	\begin{enumerate}
		\item \label{markoviandeterioration} For all $t' > t$, $S_\omega (\rho(t')) < S_\omega (\rho(t))$. 
		\item \label{eq:attractor2}For all $x'\neq y'$, $\underset{t \rightarrow \infty}{\lim} |\rho_{x'y'}(t)| = 0$.
	\end{enumerate}
\end{restatable}
The proof of the above Corollary can be found in Appendix~\ref{appendix:cor4}. Its physical meaning is clear: whenever the stochastic process generated by $L$ has a unique fixed point (with full support), such as a thermal state, coherence will eventually be destroyed. When the population finally relaxes to a stationary distribution, no coherence is left in the system. This is in stark contrast to non-Markovian covariant evolutions, where a finite fraction of coherence can always be preserved when the population reaches its fixed point $\v{\pi}$. To see this, consider a covariant channel $\E$ with the population transfer matrix $P$ defined by $P_{x'|x}=\pi_{x'}$, which transforms every initial population into $\v{\pi}$ (so its fixed point is $\v{\pi}$). The remaining parameters defining the action of $\E$ on coherence elements (see Appendix~\ref{appendix:covariant}) are
\mbox{$C^{x'|x}_{y'|y}:=\bra{x'}\E(\ketbra{x}{y})\ket{y'}$} for $\omega_{x'y'}=\omega_{xy}$. The choice $C^{x'|x}_{y'|y}=\delta_{x'x}\delta_{y'y}\sqrt{P_{x|x}P_{y|y}}$ (with $\delta_{x'x}$ denoting the Kronecker delta) guarantees that $\E$ is completely positive and results in the preservation of a fraction $\sqrt{\pi_x\pi_y}$ of every initial coherence element $\rho_{xy}$. Hence, under non-Markovian covariant dynamics some coherence can always be preserved when populations reach their fixed point $\v{\pi}$.

On the other hand, let us note that when $L$ does not have a unique fixed point, perfect \emph{Markovian} processing of coherence is possible. One such example is given by the \emph{coherence mixing} process. Consider a four-dimensional system described by Hamiltonian
\begin{equation}
\label{eq:H4}
H_4=\hbar\Omega\ketbra{1}{1}+\hbar(\Omega+\Delta)\ketbra{2}{2}+\hbar(2\Omega+\Delta)\ketbra{3}{3}.
\end{equation}
Starting with some initial values of coherence elements \mbox{$\rho_{10}(0)=|\rho_{10}(0)|$} and \mbox{$\rho_{32}(0)=|\rho_{32}(0)|$} (that belong to the same mode $\Omega$, see Fig.~\ref{fig:coh_mixing}) one may obtain an optimal evolution of coherences given by
\begin{eqnarray*}
|\rho_{10}(t)| &=& \frac{1+e^{-2\lambda t}}{2}|\rho_{10}(0)| + \frac{1-e^{-2\lambda t}}{2} |\rho_{32}(0)|,\\
|\rho_{32}(t)| &=& \frac{1-e^{-2\lambda t}}{2}|\rho_{10}(0)| + \frac{1+e^{-2\lambda t}}{2} |\rho_{32}(0)|,
\end{eqnarray*}
for some $\lambda>0$. Such a process is optimal as we have $S_{\Omega}(\rho(t))= S_{\Omega}(\rho(0))$ at all times, and it can be achieved by the following choice of $L$: 
\small
	\begin{equation*}
	L=\lambda\left[
	\begin{array}{cccc}
	-1&0&1&0\\
	0&-1&0&1\\
	1&0&-1&0\\
	0&1&0&-1
	\end{array}
	\right].
	\end{equation*}
\normalsize
Notably, we thus see that a dissipative and memoryless interaction with an environment can in some cases sustain coherence indefinitely. 

\begin{figure}
	\includegraphics[width=0.8\columnwidth]{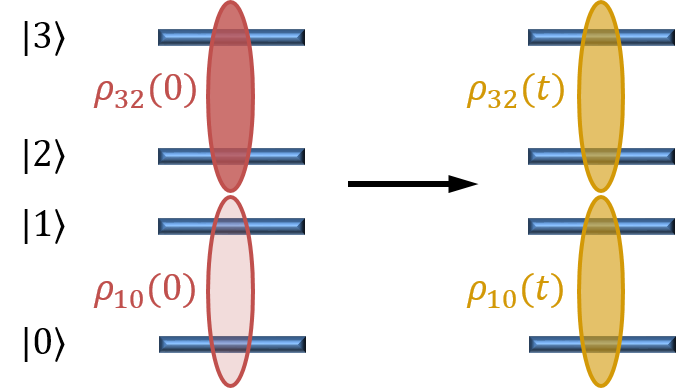}
	\caption{\label{fig:coh_mixing} \emph{Optimal coherence mixing.} Within a mode $\Omega$ (consisting of coherence elements $\rho_{10}$ and $\rho_{32}$), the optimal mixing of coherence elements can be achieved via covariant Markovian dynamics: the initial total coherence within the mode, $|\rho_{10}(0)|+|\rho_{32}(0)|$, is equal to the final total coherence, $|\rho_{10}(t)|+|\rho_{32}(t)|$, and for $t\rightarrow\infty$ we obtain $|\rho_{10}(t)|=|\rho_{32}(t)|$.}
\end{figure}

\subsection{Coherence transfer}
\label{sec:coherencetransfer}

Let us now focus on a particular type of coherence processing: coherence transfer within a mode. Consider a three-dimensional system with equidistant energy spectrum,
\begin{equation}
	H_3=\hbar\Omega(\ketbra{1}{1}+2\ketbra{2}{2}),
\end{equation}
and focus on the $\Omega$ mode composed of matrix elements $\rho_{10}$ and $\rho_{21}$. For simplicity we assume that initially only the element $\rho_{10}(0)$ of mode $\Omega$ is non-zero and we wish to maximise the final amount of coherence $\rho_{21}(t)$ (see Fig.~\ref{fig:coh_transfer}a). Similarly, consider a 4-dimensional system described by the Hamiltonian $H_4$ given in Eq.~\eqref{eq:H4}. In this case we wish to transfer coherence from $\rho_{10}$ to $\rho_{32}$ (see Fig.~\ref{fig:coh_transfer}b). Note that the only unitaries allowed by assumption \ref{ass:covariance} are energy-preserving. Hence, the only way to raise the superposition up the ladder is to extract energy from the environment. Although in both cases we deal with a mode consisting of two elements, there is an important difference. Namely, in the first case we transfer coherence between off-diagonal elements corresponding to \emph{overlapping} energy eigenstates, $\rho_{10}$ and $\rho_{21}$, whereas in the second case the transfer takes place between off-diagonal elements in \emph{non-overlapping} energy eigenstates, $\rho_{10}$ and $\rho_{32}$.

Our minimal decoherence theorem, Theorem~\ref{thm:markovian_bound}, shows that  in the first case the optimal coherence evolution achievable for a given population transfer rate $L$ is given by $d\v{c}/dt=Q\v{c}$ with	
\small
\begin{equation*}\renewcommand{\arraystretch}{1.3}
Q=\left[
\begin{array}{cc}
-\gamma_{10} & \sqrt{L_{0|1}L_{1|2}} \\
\sqrt{L_{1|0}L_{2|1}} & -\gamma_{21}
\end{array}
\right],
\end{equation*}	
\normalsize
where we introduce the coherence vector \mbox{$\v{c}:=(|\rho_{10}|, |\rho_{21}|)$}. This is a system of two first order differential equations that may be transformed into the second order differential equation
\begin{equation}
\label{eq:oscillator_equation}
\frac{d^2 c_2}{dt^2}-\mathrm{Tr}(Q) \frac{d c_2}{dt}+\det(Q) c_2=0.
\end{equation}
Since the above equation describes the motion of a damped harmonic oscillator, we see that while coherence is transferred within a mode, damping can progressively destroy it. In order to find the optimal coherence transfer we need to maximise $c_2(t)$ over all population transfer rates $L_{x'|x}$ and over all times. The solution to this problem is presented in Appendix~\ref{appendix:qutrit}, where we show that in the overlapping case, covariant Markovian evolution cannot achieve a higher coherence transfer than $\rho_{21}(t)=\rho_{10}(0)/\sqrt{2}$. This is in sharp contrast to the result for general covariant maps, where this transference task can be performed without loss of coherence, i.e., at some later $t$ we have $\rho_{21}(t)=\rho_{10}(0)$~\cite{lostaglio2015quantum}. We thus conclude that, under the covariance restriction, non-Markovianity enhances our ability to transfer coherence. 

In fact, the no-go result on Markovian transfer of coherence can also be derived directly from Corollary~\ref{corol:attractor2}. From Theorem~\ref{thm:non_markovian_bound}, any covariant process transferring coherence from an element $\rho_{10}$ to $\rho_{12}$ must have $P_{1|0}(t)>0$ and $P_{2|1}(t)>0$ at some $t>0$. Due to Markovianity this implies that the same relation holds for \emph{every} $t>0$ (see Appendix~\ref{appendix:cor4}), which gives $L_{1|0} > 0$ and $L_{2|1}>0$. However, from Corollary~\ref{corol:attractor2}, we know that perfect coherence transfer requires $L$ to have at least two zero eigenvalues. By direct inspection, one can verify that this requirement is incompatible with these two inequalities (the eigenvalues are of the form $-a\pm\sqrt{a^2-b}$ with $a>0$ and $b>0$).

\begin{figure}
	\includegraphics[width=\columnwidth]{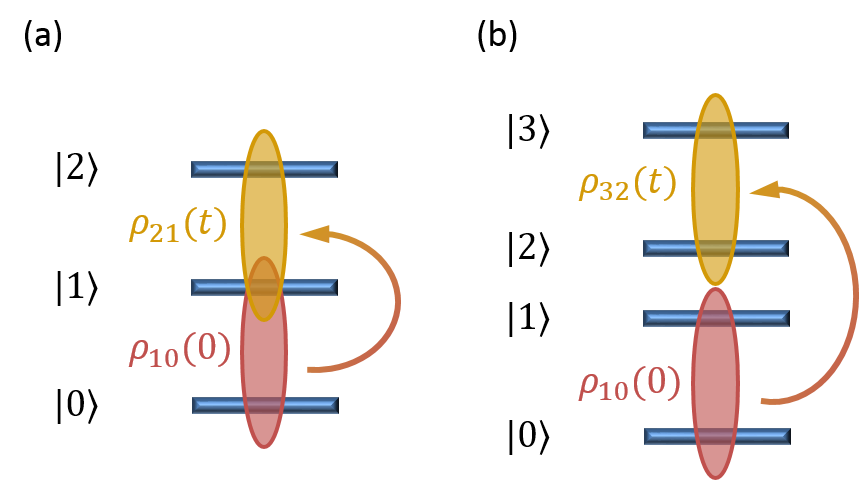}
	\caption{\label{fig:coh_transfer} (a) Coherence transfer within a mode $\Omega$ between overlapping coherence elements $\ketbra{1}{0}$ and $\ketbra{2}{1}$. (b) Coherence transfer within a mode $\Omega$ between non-overlapping coherence elements $\ketbra{1}{0}$ and $\ketbra{3}{2}$.}
\end{figure}

However, what is perhaps more surprising is that it is possible to perfectly transfer all coherence in the non-overlapping case, i.e., there exists a Markovian evolution leading to $\rho_{32}(t)=\rho_{10}(0)$ as \mbox{$t\rightarrow\infty$}. To see this, note that from Theorem~\ref{thm:markovian_bound} the optimal coherence evolution achievable for a given population transfer rate $L$ in the case of $H_4$ is again given by $d\v{c}/dt=Q\v{c}$, but this time with
\small
\begin{equation}\renewcommand{\arraystretch}{1.3}
Q=\left[
\begin{array}{cc}
	-\gamma_{10} & \sqrt{L_{0|2}L_{1|3}} \\
	\sqrt{L_{2|0}L_{3|1}} & -\gamma_{32}
\end{array}
\right],
\end{equation}
\normalsize
where the coherence vector is now $\v{c}:=(|\rho_{10}|,|\rho_{32}|)$. One can directly verify that optimal coherence transfer is achieved through the following choice of population transfer rate matrix:
\small
	\begin{equation*}
	L=\lambda\left[
	\begin{array}{cccc}
	-1&0&0&0\\
	0&-1&0&0\\
	1&0&0&0\\
	0&1&0&0
	\end{array}
	\right].
	\end{equation*}
\normalsize

This result is less surprising when we realise that the matrix $L$ that leads to a perfect transfer does not satisfy the requirements of Corollary~\ref{corol:attractor2}, i.e., it does not have a unique stationary point. Therefore, preserving coherence indefinitely within a mode through a memoryless process is possible and so, in particular, is perfect transfer within a mode. We conclude that the question of whether non-Markovianity is a resource for coherence manipulations is a subtle one that depends on the mode structure of the Hamiltonian.

\subsection{Witnessing non-Markovianity}

In this section we focus on signatures of non-Markovian dynamics. More precisely, \emph{assuming that the evolution is covariant}, we study how one can identify that the underlying dynamics is non-Markovian. We analyse two ways to achieve this: one based on monitoring the coherence of the system (which is an application of the minimal decoherence theorem), and the other on monitoring populations (which uses only the covariance condition).

\subsubsection{Coherence-based witnessing}

Let us consider a probe prepared in some state $\rho(0)$ and left in contact with an environment. What can we learn about the Markovian or non-Markovian nature of the covariant dynamics by measuring $\rho(t)$ at various times $t$? A standard approach based around the idea of ``information backflows''~\cite{breuer2016colloquium} can be applied to our scenario. Consider any distance-based measure of quantum coherence
\begin{equation}
S_D(\rho) := \min_{\sigma \in \mathcal{I}} D(\rho, \sigma),
\end{equation}
where $D$ satisfies contractivity under CPTP maps, i.e., \mbox{$D(\mathcal{E}(\rho_1),\mathcal{E}(\rho_2)) \leq D(\rho_1, \rho_2)$} if $\mathcal{E}$ is a quantum channel, and $\mathcal{I}$ is the set of states invariant under dephasing in the eigenbasis of $H$. Then, if the process is Markovian, covariance implies that for every $t'>t$ one has 
\begin{equation}
\label{eq:coherence_witness}
S_D(\rho(t')) \leq S_D(\rho(t)).
\end{equation} 
This follows directly from \mbox{$\rho(t')=\mathcal{E}_{t'-t}(\rho(t))$} and the fact that \mbox{$\rho \in \mathcal{I}$} induces \mbox{$\E(\rho) \in \mathcal{I}$}:
\small
	\begin{equation*}
	\min_{\sigma \in \mathcal{I}} D(\rho, \sigma) := D(\rho, \sigma^*) \geq D(\E(\rho), \E(\sigma^*)) \geq \min_{\sigma \in \mathcal{I}} D(\E(\rho), \sigma).
	\end{equation*}
\normalsize
Hence, violations of the inequality given in Eq.~\eqref{eq:coherence_witness} are a signature of non-Markovianity, along similar lines to Ref.~\cite{chanda2016delineating}.

An alternative approach~\cite{wolf2008assessing} assumes that we only known the initial preparation $\rho(0)$ and the final state $\rho(t)$ at a unique time $t>0$.  As our previous example of coherence transfer between overlapping coherence elements illustrates, sometimes such a single ``snapshot" can be enough to deduce non-Markovianity. To simplify the argument let us assume that the probe is a single qubit and that the dynamics has some known fixed point $\rho(\infty)$ with occupations \mbox{$\v{\pi}=(\pi,1-\pi)$} and zero coherence. What can we learn from a single snapshot of qubit dynamics? The situation is analogous to that of the elementary example from Sec.~\ref{sec:elementaryexample}. The initial state is given by Eq.~\eqref{eq:initialstate}. If the evolution is Markovian, Theorem~\ref{thm:markovian_bound} applies, leading to an optimal process described by Eq.~\eqref{eq:elementary} with $\gamma = (\abs{L_{0|0}} +\abs{L_{1|1}})/2$. Solving the equations for the optimal process one obtains
\begin{equation}
\label{eq:qubitmarkovianbound}
|c(t)| = \sqrt{\frac{p(t) - \pi}{p(0) - \pi}}|c(0)|,
\end{equation}
as shown in Fig.~\ref{fig:example} and illustrated for example initial and stationary states by the dashed trajectories of Fig.~\ref{fig:qubit_markovian_vs_non}. 

If the observed final state $\rho(t)$ lies outside the dashed region, we can infer that we are witnessing non-Markovianity. Note that this includes cases where classical information alone, i.e., measurement of $p(t)$, would be inconclusive by itself (in Fig.~\ref{fig:qubit_markovian_vs_non}, this is the case when $p(t) \leq \pi$). It also includes dynamics that, despite satisfying $|c(t)| < |c(0)|$, are still incompatible with a Markovian process as they preserve too much coherence.

\begin{figure}
	\includegraphics[width=\columnwidth]{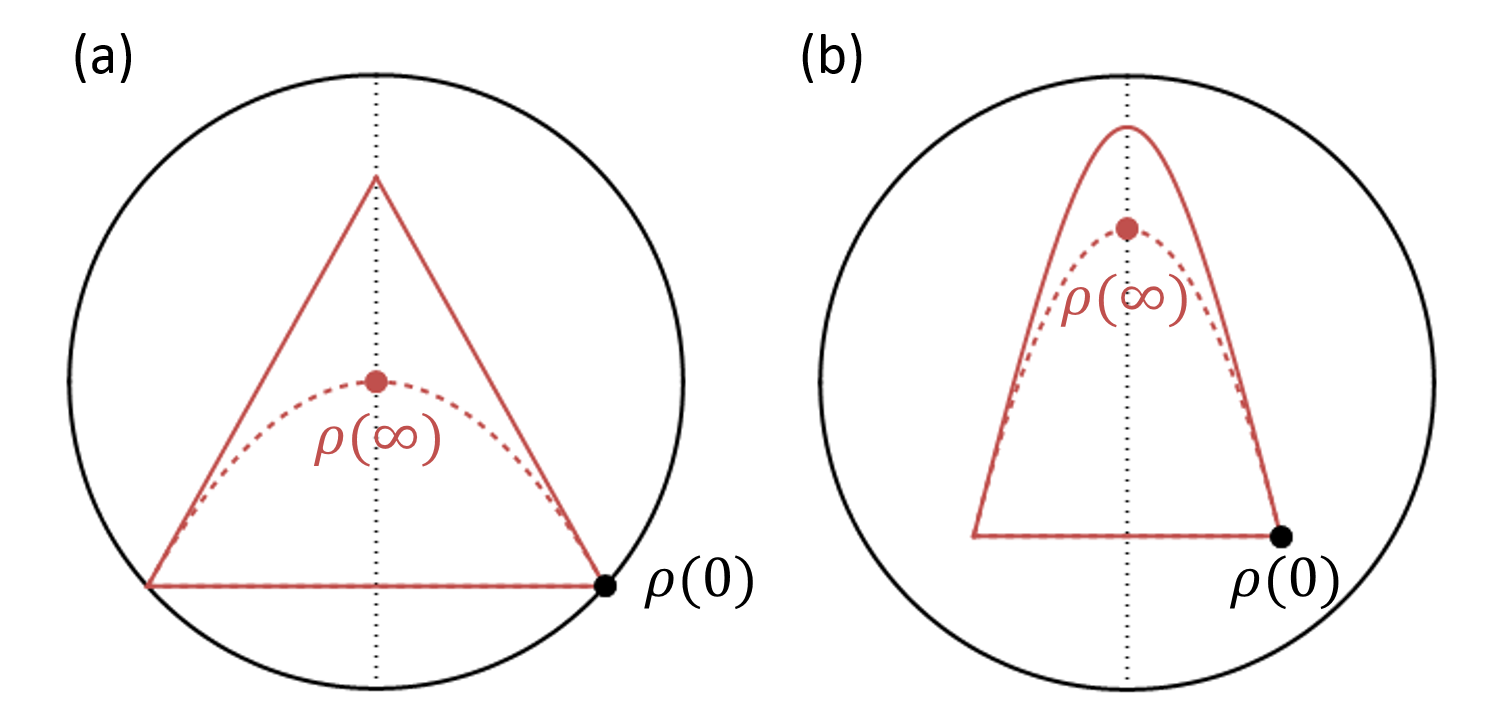}
	\caption{\label{fig:qubit_markovian_vs_non} \emph{Qubit covariant dynamics: Markovian vs. non-Markovian.} (a) Initial state with $p(0)=1/6$ and \mbox{$c(0)=\sqrt{5}/6$}, and $\rho(\infty)$ such that $\v{\pi}=(1/2,1/2)$. (b) Initial state with $p(0)=1/4$ and $c(0)=1/4$, and $\rho(\infty)$ such that \mbox{$\v{\pi}=(3/4,1/4)$}. The dashed lines show the maximum coherence preservation possible with Markovian covariant dynamics with a given fixed point $\rho(\infty)$; the solid lines show the maximum coherence preservation possible for general covariant operations with a fixed state given by $\rho(\infty)$.}
\end{figure}

\subsubsection{Population-based witnessing}

Even though complete tomographic knowledge about the final state $\rho(t)$ gives more powerful ways to identify non-Markovianity, owing to the particular structure of covariant maps, sometimes knowledge of the population dynamics is sufficient. As described in Eq.~\eqref{eq:motionforpopulation}, a covariant Markovian channel induces a stochastic matrix $P$ on the vector of populations, so that $\v{p}(t) = P \v{p}(0)$, where $P = e^{L t}$. Technically, one can say that the stochastic matrix $P$ must be \emph{embeddable}, which means that it is induced by exponentiation of a generator $L$. As not all stochastic matrices $P$ can be generated this way, embeddability of the population dynamics matrix gives a necessary condition for the channel to be Markovian. In particular, one can use the following known result~\cite{runnenburg1962elfving,minc1988nonnegative,davies2010embeddable}:

\begin{thm}\label{thm:embeddable}
	The eigenvalues $\{\lambda\}$ of a $d\times d$ embeddable stochastic matrix $P$ must satisfy $\lambda=r e^{i\phi}$ with \mbox{$-\pi\leq\phi\leq\pi$}, \mbox{$0\leq r\leq r(\phi)$} and \mbox{$r(\phi)=e^{-|\phi|\tan(\pi/d)}$}. In other words the eigenvalues are bounded to the region of the complex plane that lies inside the curve $x(\phi)+iy(\phi)$ with
		\begin{equation}
		\label{eq:embeddable_region}
		x(\phi)=e^{-|\phi|\tan\frac{\pi}{d}} \cos\phi,\quad
		y(\phi)=e^{-|\phi|\tan\frac{\pi}{d}} \sin\phi.		
		\end{equation}
\end{thm}
\noindent For the convenience of the reader we present the proof in Appendix~\ref{appendix_embeddable}.

Given a population transfer matrix $P$ acting on a \mbox{qu-$d$-it} system, one can use Theorem~\ref{thm:embeddable} to verify whether any of its eigenvalues lie outside of the ``embeddability region''.\footnote{If only $\v{p}(0)$ and $\v{p}(t)$ are known, this requires the study of all stochastic $P$ such that $P\v{p}(0) = \v{p}(t)$.} In order to understand how restrictive this condition is we compare the embeddable region, specified by Theorem~\ref{thm:embeddable}, with the region occupied by the eigenvalues of generic $d\times d$ stochastic matrices, specified by the Karpelevi\u{c} theorem~\cite{karpelevich1951characteristic,ito1997new}. We present this comparison in Fig.~\ref{fig:stoch_eig} for several small values of $d$. Whereas for small dimensions a large class of covariant operations can be deemed non-Markovian  by simply analysing the dynamics of populations, the bound becomes less tight for higher dimensional systems, illustrating the relevance of the previous considerations involving coherence. Interestingly, some important transformations are necessarily non-Markovian in any dimension. As an example consider ``probabilistic rigid translations'' defined by a stochastic map \mbox{$T(q)=(1-q)\iden+qP$} with $q\in(0,1]$ and $P$ a cyclic permutation, i.e., $P_{i+1|i} = 1$ for $i=0,...,d-2$ and  $P_{0|d-1} =1$ (or $P_{i-1|i} = 1$ for $i=1,...,d-1$ and $P_{d-1|0} =1$). Since $P^d=\iden$, one of the eigenvalues of $T(q)$ is \mbox{$1-q+qe^{2\pi i/d}$}, which lies on the edge connecting the points $(1,0)$ and \mbox{$(\cos \frac{2\pi}{d},\sin \frac{2\pi}{d})$}. However, as can be verified using Eq.~\eqref{eq:embeddable_region}, for $d\geq 3$, the eigenvalues of embeddable stochastic matrices will always lie below this edge, and hence ``probabilistic rigid translations'' cannot be achieved using Markovian dynamics.

\begin{figure}[t!]\centering
	\includegraphics[width=0.8\columnwidth]{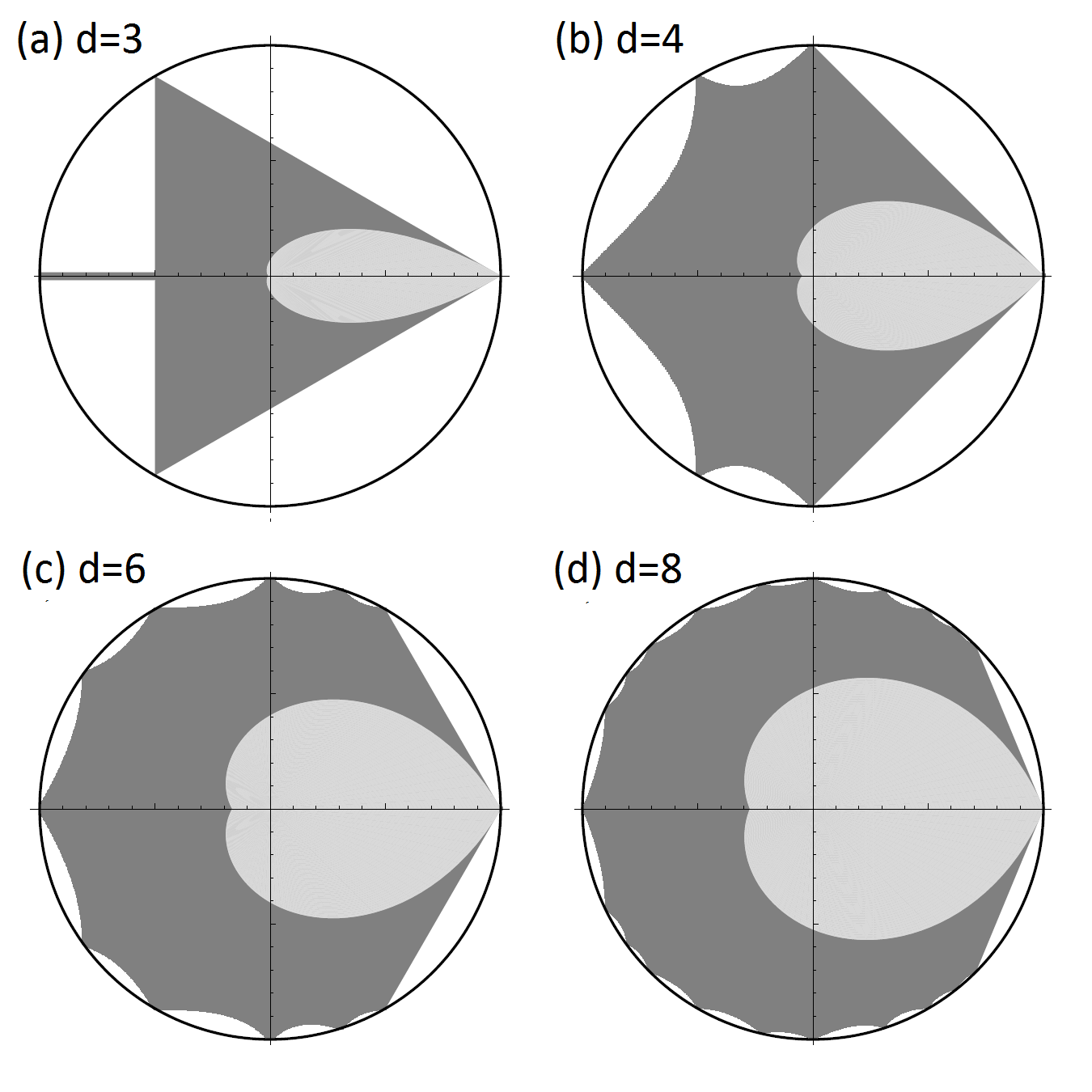}
	\caption{\label{fig:stoch_eig} \emph{Eigenvalues of stochastic matrices.} The eigenvalues of a $d\times d$ stochastic matrix all lie within the unit circle on the complex plane, independently of $d$. For a given $d$, points corresponding to the eigenvalues of a stochastic $d\times d$ matrix are given by the Karpelevi\u{c} theorem~\cite{karpelevich1951characteristic,ito1997new} and are depicted in dark grey. Points that correspond to the eigenvalues of an embeddable stochastic $d\times d$ matrix, specified by Theorem~\ref{thm:embeddable}, are depicted in light grey.}
\end{figure}

\subsection{Resource theory of thermodynamics}

Despite a great amount of work pursued within the so-called resource-theoretic formulation of quantum thermodynamics (see Ref.~\cite{goold2016role, vinjanampathy2016quantum} and references therein), its relation to the standard formalism of master equations and thermalisation models has not been clarified. This has generated confusion regarding the scope of the results derived within the resource theory and their relevance for applications~\cite{halpern2017toward}. In this section we present a unified picture that relates both formalisms, and show explicitly how the technical machinery of open quantum systems can be used to strengthen the resource-theoretic approach in physically relevant scenarios.  
	
\subsubsection{From generalised thermal operations to standard thermalisation models}
	
Within the resource theory of thermodynamics one studies the possible dynamics of quantum systems induced by the restricted set $\{\E\}$ of quantum channels known as \emph{thermal operations}~\cite{janzing2000thermodynamic}. However, all constraints on the allowed transformations derived so far emerge from two core properties: covariance of $\E$ with respect to time-translations, as given in Eq.~\eqref{eq:covariance}, and the presence of a thermal fixed point, i.e., $\mathcal{E}(\tau) = \tau$ with $\tau = e^{-\beta H}/ \tr{}{e^{-\beta H}}$ being the thermal Gibbs state at inverse temperature $\beta$~\cite{lostaglio2015quantum}. The set of channels satisfying these two properties has been called \emph{generalised thermal operations} (GTOs) in Ref.~\cite{cwiklinski2015limitations}. 

\begin{figure}[t!]\centering
	\includegraphics[width=\columnwidth]{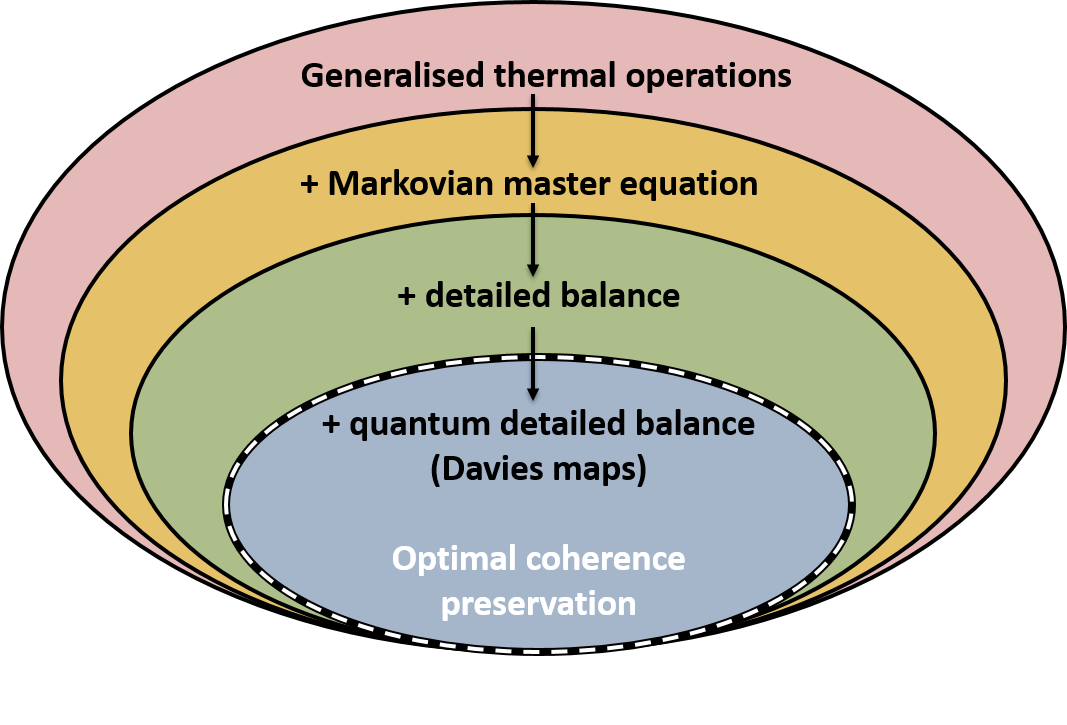}
	\caption{\label{fig:structure} \emph{Families of thermodynamic quantum channels.} By incorporating assumptions of Markovianity and quantum detailed balance into the set of generalised thermal operations, one obtains the set of Davies maps. These are precisely those which achieve optimal coherence transformation in Theorem~\ref{thm:markovian_bound}.}
\end{figure}

We will now argue that GTOs, in a precise sense, are a generalisation of a class of master equations whose properties are commonly assumed or derived in the study of thermalisation processes. Fig.~\ref{fig:structure} illustrates the overall structure of the sets of quantum channels we consider. More precisely, it shows how one can move from the resource theory of thermodynamics to a standard open quantum system description by adding certain physical restrictions. The largest class of operations represents GTOs. The dynamics of coherence under these channels is limited by the constraints of Theorem~\ref{thm:non_markovian_bound}, with transition probabilities $P_{x'|x}$ satisfying $P \v{\pi} = \v{\pi}$, where $\v{\pi}$ is a vector of thermal occupations (eigenvalues of $\tau$), i.e., $\pi_x\propto e^{-\beta \hbar\omega_x}$. GTOs can then naturally be restricted to the subset that admits a realisation through a Markovian master equation, i.e., satisfying Assumption~\ref{ass:markov}. This set presents stronger constraints on processing coherence, in the form of Theorem~\ref{thm:markovian_bound}, where $P \v{\pi} = \v{\pi}$ translates into the condition $L \v{\pi} = \v{0}$ on the transition rates $L_{x'|x}$. The next step is to recognise that the relation \mbox{$L\v{\pi} = \v{0}$} is itself simply a weaker form of the so-called detailed balance condition, 
\begin{equation}
\label{eq:detailedbalance}
L_{x'|x} \pi_x = L_{x|x'} \pi_{x'}.
\end{equation}
This stronger condition is satisfied, for example, by Davies maps, which describe standard thermalisation models whose microscopic derivation involves large thermal baths and weak couplings~\cite{davies1974markovian, roga2010davies}. In our formalism, detailed balance can be implemented simply by restricting the allowed transition rates in Theorem~\ref{thm:markovian_bound}. We also note that Davies maps are covariant Markovian channels satisfying an even stronger form of Eq.~\eqref{eq:detailedbalance} known as \emph{quantum detailed balance}~\cite{roga2010davies}. To complete the connection between the resource-theoretic and master equation formalisms we make the following observation: among all detailed balanced GTOs that admit a master equation description, those that transform coherence optimally are Davies maps. We will formally prove this by showing that optimal coherence transformations require quantum detailed balance and hence restrict us to the smallest set shown in Fig.~\ref{fig:structure}.  

To sum up, GTOs can be understood as a generalisation of Davies maps in which the following conditions are relaxed:
\begin{enumerate}
	\item The map does not necessarily admit a Markovian master equation description, i.e., Assumption~\ref{ass:markov} is dropped.
	\item Quantum detailed balance is relaxed to the minimal condition that the thermal state is a fixed point of the dynamics.
\end{enumerate}
In the remainder of this section we first prove the already mentioned connection between quantum detailed balance and optimality of coherence preservation. We then illustrate how additional physical restrictions on the resource theory, identified in Fig.~\ref{fig:structure}, allow us to obtain stronger constraints on the allowed transformations.

\subsubsection{Quantum detailed balance and optimal coherence processing}
	
The dynamics generated by the dissipator $\L$ satisfy the quantum detailed balance condition when~\cite{roga2010davies}
\begin{equation}
\label{eq:qdb}
	\tr{}{\mathcal{L}\left(A \tau\right)B} =\tr{}{\mathcal{L}\left(\tau B\right)A}, 
\end{equation}
for all $d \times d$ matrices $A$, $B$, with $\tau$ denoting a thermal Gibbs state. We will now prove that, under the assumptions of Theorem~\ref{thm:markovian_bound}, with the Markovian phase-matching condition holding and population transition rates satisfying the detailed balance condition [Eq.~\eqref{eq:detailedbalance}], the maps that transform coherence optimally satisfy Eq.~\eqref{eq:qdb}, and hence are Davies maps. 

Due to linearity, quantum detailed balance holds if and only if it holds for all $A$ and $B$ of the form $A=\ketbra{x}{y}$ and $B=\ketbra{y'}{x'}$, i.e.,
\footnotesize
\begin{equation}
\label{eq:quant_dbc}
e^{-\beta\hbar\omega_y}\matrixel{x'}{\L\left(\ketbra{x}{y}\right)}{y'} =e^{-\beta\hbar\omega_{y'}}\matrixel{x}{\L\left(\ketbra{x'}{y'}\right)}{y}^*. 
\end{equation}\normalsize
Note that, due to covariance, we only need to consider $\ketbra{x}{y}$ and $\ketbra{x'}{y'}$ belonging to the same mode, since all other terms vanish. 

For mode zero ($x=y$ and $x'=y'$), Eq.~\eqref{eq:quant_dbc} simply yields Eq.~\eqref{eq:detailedbalance} and thus holds by assumption. For non-zero modes we need to use the expression for the optimal channel. This is given by (see Appendix~\ref{appendix:proof} for details):
\begin{subequations}
\begin{eqnarray}
\L(\cdot) &=& \A(\cdot) - \tfrac{1}{2}\{\A^\dag(\iden), \cdot\}, \\
A^{x'|x}_{y'|y} &=& \vartheta_{x'x}(0) \vartheta^*_{y'y}(0)  \sqrt{L_{x'|x}L_{y'|y}}. 
\end{eqnarray}
\end{subequations}
If either $x'\neq x$ or $y'\neq y$ the above can be simplified as [see Eq.~\eqref{eq:anticomm_els}]
\begin{equation}
\matrixel{x'}{\L(\ketbra{x}{y})}{y'}=\matrixel{x'}{\A(\ketbra{x}{y})}{y'}.
\end{equation}
Using this equation and the expression for the optimal $\A$, it is now straightforward to show that Eq.~\eqref{eq:quant_dbc} holds. To complete the proof we need to show that Eq.~\eqref{eq:quant_dbc} also holds when $x=x'$ and $y=y'$. This is equivalent to \mbox{$\matrixel{x}{\L(\ketbra{x}{y})}{y}$} being real, which can be easily verified.

\subsubsection{Strengthening the resource theory constraints}

Finally, we demonstrate how our framework can be used to derive new and stronger bounds on the processing of coherence under thermal operations when additional physical constraints hold. We first consider the simplest case of a qubit system. Then the hierarchy of Fig.~\ref{fig:structure} simplifies to three sets: general GTOs, GTOs admitting a Markovian master equation, and Davies maps. The fixed thermal state is specified by $\v{\pi} = (\pi, 1-\pi)$ with $\pi = 1/(1+e^{-\beta \hbar \omega})$, where $\omega$ is the relevant transition frequency. Recall that we denote the off-diagonal element of the density operator in the energy eigenbasis by $c(t)$ and the ground state population by $p(t)$. In Ref.~\cite{cwiklinski2015limitations,lostaglio2015quantum} it was shown that for transformations induced by GTOs the following tight bound holds:
\begin{equation}
\label{eq:nmto}
|c(t)| \leq \frac{\sqrt{(p(t)-q_\beta(0))(p(0)-q_\beta(t))}}{|p(0) - q_\beta(0)|}|c(0)|,
\end{equation} 
where \mbox{$q_\beta(t):=(1-p(t))e^{\beta \hbar \omega}$}. The bound is marked with solid lines in Fig.~\ref{fig:qubit_markovian_vs_non}. This can be obtained directly from Theorem~\ref{thm:non_markovian_bound} by imposing $P \v{\pi} = \v{\pi}$~\cite{lostaglio2015quantum}, and was shown to be achievable under thermal operations in Ref.~\cite{cwiklinski2015limitations}.

Employing the relation given in Eq.~\eqref{eq:qubitmarkovianbound}, one obtains a tighter bound for GTOs that admit a Markovian master equation, namely:
\begin{equation}
|c(t)| \leq \sqrt{\frac{p(t) - q_\beta(t)}{p(0) - q_\beta(0)}}|c(0)|.
\end{equation} 
This bound is plotted with dashed lines in Fig.~\ref{fig:qubit_markovian_vs_non}. From the discussion presented in this section we know that processes achieving this bound are Davies maps. Moreover, we see that for any $\beta \neq \infty$ and any initial state carrying quantum coherence, the saturation of the thermal operation (or GTO) bound always requires non-Markovianity.\footnote{Also note that at zero temperature (a situation studied in Ref.~\cite{narasimhachar2015low}) the two regions coincide, i.e., any transformation that can be achieved by the full set of thermal operations can be also achieved by Markovian processes.}

For higher dimensional systems the relation $L \v{\pi} = \v{0}$, which captures the irreversibility of thermal operations, allows one to find temperature-dependent bounds on the coherence transport rates $t^{x'|x}_{y'|y}$. More precisely, taking $\pi_x\propto e^{-\beta \hbar\omega_x}$ one can show that (see Appendix~\ref{appendix:lbound} for details)
\begin{equation}
\label{eq:transport_bound}
t^{x'|x}_{y'|y}\leq \gamma_{x'y'} \min \left\{e^{-\beta \hbar \omega_{x' x}}, 1\right\},	
\end{equation}
so that transport processes responsible for	moving coherence up in energy are exponentially damped by a Gibbs factor $e^{-\beta \hbar \omega_{x' x}}$. This becomes more pronounced when one additionally assumes the detailed balance condition for the population dynamics, i.e., \mbox{$L_{x'|x}=L_{x|x'}e^{-\beta\hbar\omega_{x'x}}$}, resulting in asymmetry between transport rates:
\begin{equation}
t^{x'|x}_{y'|y}\leq t^{x|x'}_{y|y'} \min \left\{e^{-\beta \hbar \omega_{x' x}}, 1\right\}.
\end{equation}
These relations are the analogue at the level of master equations of the results derived for GTOs in Ref.~\cite{lostaglio2015quantum}.

\section{Conclusions}
\label{sec:conclusions}

In this work we have attempted to unify two rather contrasting concepts. On the one hand, Holevo introduced an approach to characterise the generators of dynamics compatible with a given symmetry~\cite{holevo1993note, holevo1995structure}. On the other hand, recent theoretical works from quantum information and the theory of reference frames present the \emph{lack} of symmetry of a quantum state as a consumable resource when dynamics are restricted by a symmetry principle~\cite{bartlett2007reference,marvian2014extending}. In the specific case of symmetry under time translation considered in this paper, this resource coincides with quantum coherence in the basis of the system Hamiltonian~\cite{lostaglio2015quantum, marvian2016quantify}. 

We have investigated what are the fundamental limitations on the processing of such coherences. Our results are derived using purely the underlying symmetry and the assumption of Markovianity, without any reference to the specific features of a particular model. This yields general bounds connecting the evolution of populations and coherences, from which a wide variety of further results are easily obtained. We find a $d$-dimensional generalisation of the classic inequality $T_2 \leq 2 T_1$ that relates the relaxation time $T_1$ and the decoherence time $T_2$ of a qubit. Highlighting the relationship between ergodicity and the preservation of coherence enables us to study the role of non-Markovianity as a resource for coherence processing. It also raises the possibility of engineering a symmetric dissipative interaction to have multiple fixed points, with the aim of protecting coherent resources (in a similar spirit to Ref.~\cite{lin2013dissipative}). By providing explicit examples we show how non-Markovian transformations can enhance coherence processing under symmetric dynamics, motivating the utility of a resource-theoretic formulation of non-Markovian processes~\cite{rivas2014quantum}. We also present methods for witnessing such non-Markovian behaviour through the dynamics of coherences and populations. These are based on the underlying symmetry of the dynamics and, as such, illustrate how symmetry reasoning can simplify the detection of non-Markovianity. We point out that the possibility of simultaneously saturating our bound for all coherence elements, or of weakening the Markovian phase-matching condition, remains open for future investigations.

We have also explicitly connected the resource theory approach to thermodynamics with the well-established master equation formalism. As well as clarifying the physical scope of the abstract resource theory, this paves the way for the use of the open quantum systems dynamics toolkit to tackle resource-theoretic questions, as our bounds illustrate. Conversely, new insights can be obtained into well-established models of thermalisation, as demonstrated by the optimal coherence preservation properties of Davies maps. More generally, we have presented evidence that our approach allows one to establish fruitful links between two formalisms used to study thermodynamics that were principally developed independently.

In this study we have focused on time-translation symmetry, but all the main ingredients can in fact be generalised to dynamics that are symmetric with respect to an arbitrary group $G$. The result of Holevo~\cite{holevo1993note}, the so-called resource theory of asymmetry~\cite{marvianthesis} and the harmonic analysis tools~\cite{marvian2014modes} used throughout this paper all apply to general groups. Hence one should be able to derive relations for the evolution of a generalised ``coherence'' for different observables. Resource states in the generalised theory are those which are \emph{asymmetric} with respect to a group $G$, i.e., they evolve non-trivially under its action~\cite{marvian2013asymmetry, marvian2014extending}. In our study of time-translation symmetry we have taken $G=U(1)$, and such states possess quantum coherence relative to the basis defined by the Hamiltonian. If we were to instead take $G=SU(2)$, i.e., spherically symmetric dynamics, then a state that is sensitive to rotations (a superposition of different angular momentum eigenstates) would constitute a resource. We thus hope that the results we have presented here suggest a general resource-theoretic approach for studying the consequences of symmetry within the master equation formalism.

\subsection*{Acknowledgments}
    We thank A.~Ac\'in, B.~Bylicka, G.~Canevari, K.~Hovhannisyan, F.~Wirth, K.~\.Zyczkowski for useful discussions. ML acknowledges financial support  
    from the Spanish MINECO (Severo Ochoa
    SEV-2015-0522 and project QIBEQI FIS2016-80773-P), Fundacio Cellex, and Generalitat de Catalunya (CERCA
    Programme and SGR 875). KK acknowledges support from the ARC via the Centre of Excellence in Engineered Quantum Systems (EQuS), project number
    CE110001013. This work was supported by EPSRC and in part by
    COST Action MP1209.	All authors contributed equally to this work.
	    
\bibliography{Bibliography_thermodynamics}

\appendix

\section{Covariant maps -- a walk-through}
\label{appendix:covariant}

\subsection{Definition}

Consider a $d$-dimensional system with non-degenerate Hamiltonian $H = \sum_{x=0}^{d-1} \hbar\omega_x \ketbra{x}{x}$. A completely positive map $\E$ is called \textit{covariant} with respect to time translations generated by $H$ (also known as \textit{phase-insensitive}) when
\begin{equation}
	\mathcal{E}\left(e^{-iHt}\rho e^{iHt}\right) = e^{-iHt} \mathcal{E}(\rho)e^{iHt}.
\end{equation}
holds for all $\rho$ and $t$. Note that this is a special case of covariance with respect to a general (compact) group $G$ \cite{marvian2014extending}. Equivalent characterisations of covariant maps can be found in Ref.~\cite{marvian2016quantify}, Section~IIIA.

\subsection{Modes of coherence}
\label{app_sec:modesofcoherence}

Recall that, as in the main text, the set $\{\omega\}$ consisting of all differences between eigenfrequencies of $H$ is known as the \textit{Bohr spectrum} of $H$. A covariant map can be decomposed according to its action on the \emph{modes} of a state~\cite{marvian2014modes} (an analogue of the Fourier component, but for quantum states). The mode structure is defined by the Bohr spectrum of $H$. More precisely, the state \mbox{$\rho = \sum_{x,y} \rho_{xy} \ketbra{x}{y}$} can be written in the form
\begin{equation}
\label{eq:modes}
\rho = \sum_{\omega} \rho^{(\omega)},
\end{equation}
where
\begin{equation}
\label{eq:modesofcoherencerho}
\rho^{(\omega)} = \sum_{\mathclap{\substack{x,y \\ \omega_{xy} = \omega}}} \rho_{xy} \ketbra{x}{y}=: \sum^{(\omega)}_{\substack{x,y}} \rho_{xy} \ketbra{x}{y}.
\end{equation}
Here we have introduced the symbol $\sum^{(\omega)}$ to indicate the sum over indices $x,y$ such that $\omega_{xy}=\omega$ (recall that $\omega_{xy}=\omega_x-\omega_y$). The operators $\rho^{(\omega)}$ are called \textit{modes of coherence} of the state $\rho$. Now, if $\E$ is a covariant operation such that \mbox{$\E(\rho)=\sigma$} then
\begin{equation}
\label{eq:covariancemodes}
\E\left(\rho^{(\omega)}\right) = \sigma^{(\omega)} \quad \quad \forall \omega.
\end{equation}
The converse is also true (see Ref.~\cite{marvian2016quantify}, Proposition~6). In other words, each mode $\rho^{(\omega)}$ of the initial state is independently mapped to the corresponding mode $\sigma^{(\omega)}$ of the final state if and only if the mapping is achieved via a covariant operation. 

We can now conveniently parametrise a covariant map $\E$ in the following way. First, let us define the action of $\E$ on diagonal (i.e., mode zero) matrix elements by
\begin{equation}
\label{eq:action_diag}
\E(\ketbra{x}{x})=\sum_{x'=0}^{d-1} P_{x'|x} \ketbra{x'}{x'},
\end{equation}
where, if $\mathcal{E}$ is trace-preserving, the coefficients are the entries $P_{x'|x}$ of a stochastic matrix $P$ ($P_{x'|x}\geq 0$ and \mbox{$\sum_{x'}P_{x'|x}= 1$}). This can naturally be interpreted as the population transfer matrix, i.e., the transition matrix between energy eigenstates. To see this more clearly, let us introduce the \emph{vector of populations} $\v{p}$ that describes the zero mode of $\rho$, i.e., its components are given by $p_x=\rho_{xx}$. The transformation of the zero mode under $\E$ is then described by the transformation $\v{p}\mapsto P\v{p}$.

The action of $\E$ on an off-diagonal matrix element $\ketbra{x}{y}$ can be parametrised as follows:
\begin{equation}
\label{eq:action_offdiag}
\E(\ketbra{x}{y})=\sum^{\om{x}{y}}_{x',y'}  C^{x'|x}_{y'|y} \ketbra{x'}{y'}.
\end{equation}
The coefficients $C^{x'|x}_{y'|y}$ describe how much the initial coherence $\ketbra{x}{y}$ contributes to the final coherence $\ketbra{x'}{y'}$. Hermiticity of the final state imposes \mbox{$C^{x'|x}_{y'|y}=(C^{y'|y}_{x'|x})^*$}. Note that formally $P_{x'|x}$ can be thought of as $C^{x'|x}_{x'|x}$.

As an example we now look at a qubit system, which without loss of generality can be described by the Hamiltonian $H=\hbar \Omega \ketbra{1}{1}$.
\begin{ex}
	\label{ex:qubit1}
	The state $\rho$ can be decomposed into three modes consisting of the following matrix elements:
	\begin{eqnarray*}
		\rho^{(0)}&:&\{\ketbra{0}{0},\ketbra{1}{1}\},\\
		\rho^{(\Omega)}&:&\{\ketbra{1}{0}\},\\
		\rho^{(-\Omega)}&:&\{\ketbra{0}{1}\}.
	\end{eqnarray*}
	As a covariant map does not mix modes, the action of $\E$ on $\rho$ is given by
	\begin{eqnarray*}
		\E(\ketbra{0}{0})&=&P_{0|0}\ketbra{0}{0}+P_{1|0}\ketbra{1}{1},\\
		\E(\ketbra{1}{1})&=&P_{1|1}\ketbra{1}{1}+P_{0|1}\ketbra{0}{0},\\
		\E(\ketbra{1}{0})&=&C^{1|1}_{0|0}\ketbra{1}{0},\\
		\E(\ketbra{0}{1})&=&C^{0|0}_{1|1}\ketbra{0}{1}.
	\end{eqnarray*}
	Since \mbox{$P_{1|0}=1-P_{0|0}$}, \mbox{$P_{0|1}=1-P_{1|1}$} and \mbox{$C^{1|1}_{0|0}=(C^{0|0}_{1|1})^*$}, a general covariant qubit map is fully specified by two transition probabilities, $P_{0|0}$ and $P_{1|1}$, and a complex number $C^{0|0}_{1|1}$.
\end{ex}

\subsection{Structure of the Choi-Jamio\l kowski state}
\label{app_sec:Choi_structure}

Until now we have described the conditions for $\E$ to be covariant. However, in order to represent a physical transformation $\E$ must also be completely positive (CP) and, since we look at deterministic transformations, we take $\E$ to be trace-preserving (TP). The latter property is automatically satisfied given the mode structure and the fact that $P$ is a stochastic matrix. To see this, note that the trace of the final state can be written as
\small
\begin{equation*}
\tr{}{\E(\rho)} = \tr{}{\E(\rho)^{(0)}} = \tr{}{\E(\rho^{(0)})} = \sum_x(P\v{p})_x=1,
\end{equation*}
\normalsize
where we have used Eq.~\eqref{eq:covariancemodes} and the fact that a stochastic matrix maps the space of probability vectors onto itself.

To enforce the CP condition, we recall the Choi-Jamio\l{}kowski isomorphism~\cite{jamiolkowski1972linear,choi1975completely,bengtsson2006geometry} that maps a quantum channel $\E$ into the state 
\begin{equation}
J[\E]:=[\E \otimes \mathcal{I}]\left(\ketbra{\phi^+}{\phi^+}\right),
\end{equation}
where \mbox{$\ket{\phi^+} =\sum_{x=0}^{d-1} \ket{xx}/\sqrt{d}$} is the maximally entangled state, and $\mathcal{I}$ denotes the identity superoperator. The requirement of $\E$ to be CP is equivalent to the Choi operator $J[\E]$ being positive semidefinite. Writing out $J[\E]$ explicitly we have
\begin{eqnarray}
	\label{eq:choi}
	J[\E]&=& \frac{1}{d}\sum_{x,y} \sum^{\om{x}{y}}_{x',y'} C^{x'|x}_{y'|y}\ketbra{x'}{y'} \otimes \ketbra{x}{y}\nonumber\\
	&=&\frac{1}{d} \sum_{x',x} \sum^{\omp{x}{x}}_{y',y} C^{x'|x}_{y'|y}\ketbra{x'x}{y'y},
\end{eqnarray}
where we have rearranged the expression to emphasise the block-diagonal structure. Note that $J[\mathcal{E}]$ is block-diagonal with respect to the eigenbasis of $ H \otimes \iden -\iden \otimes H^*$.  

Each block consists of matrix elements $C^{x'|x}_{y'|y}$ for which \mbox{$\omega_{x'x} = \omega_{y'y}=\omega$} and can thus be labelled by $\omega$ (see Fig.~\ref{fig:choi_state}). The diagonal of each block $\omega$ consists of population transfer coefficients $P_{x'|x}$ with $\omega_{x'x}=\omega$, i.e., it describes the population transfers between energy levels differing by $\hbar\omega$ in energy. Off-diagonal elements in the $\omega=0$ block, $C^{x|x}_{y|y}$, correspond to the fraction of the initial coherence term $\rho_{xy}$ that is preserved in the final state (modulo adding phases); off-diagonal elements in blocks with $\omega\neq 0$, $C^{x'|x}_{y'|y}$, describe the transfer of the initial coherence term $\rho_{xy}$ into the final coherence term $\sigma_{x'y'}$.

\begin{figure}
	\includegraphics[width=\columnwidth]{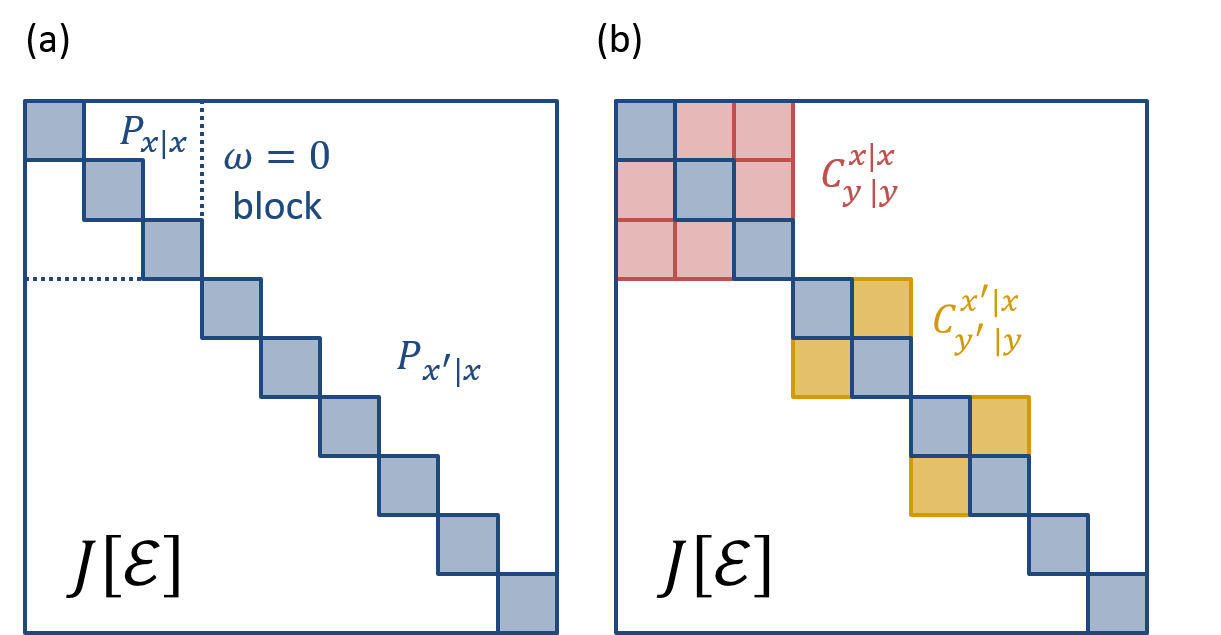}
	\caption{\label{fig:choi_state} \emph{The block-diagonal structure of the Choi state.} An example of the Choi state $J[\E]$ of a covariant map $\E$ for a qutrit system described by an equidistant Hamiltonian, i.e., $H=\hbar\Omega(\ketbra{1}{1}+2\ketbra{2}{2})$. (a)~The diagonal terms of $J[\E]$ (in blue) describe the evolution of populations, i.e., the transition rates $P_{x'|x}$ between energy eigenstates $\ketbra{x}{x}$ and $\ketbra{x'}{x'}$. (b)~The off-diagonal terms of $J[\E]$ describe the preserved ``fraction'' $C^{x|x}_{y|y}$ of coherence term $\ketbra{x}{y}$ (in red), and the amount $C^{x'|x}_{y'|y}$ of coherence transferred (in yellow) between coherence terms $\ketbra{x}{y}$ and $\ketbra{x'}{y'}$.}
\end{figure}

\subsection{Positivity of the Choi-Jamio\l kowski state}
\label{app_sec:positivitychoi}

Owing to the block-diagonal structure of $J[\E]$, positivity is equivalent to the positivity of each block. A necessary condition for the positivity of block $\omega$ is that for all $x,y$ and $x',y'$ within, one has
\begin{equation}
\label{eq:opt_coh}
\left|C^{x'|x}_{y'|y}\right|\leq \sqrt{P_{x'|x}P_{y'|y}},
\end{equation}
i.e., the magnitude of the off-diagonal term is constrained by the corresponding diagonal terms.  Now note that from Eq.~\eqref{eq:action_offdiag} and the triangle inequality, we have
\begin{equation}
|\sigma_{x'y'}| \leq \sum^{(\omega_{x'y'})}_{x,y}\left|C^{x'|x}_{y'|y}\right| |\rho_{xy}|.
\end{equation}
Using the above together with Eq.~\eqref{eq:opt_coh} immediately yields Eq.~\eqref{eq:non_markovian_bound} from Theorem~\ref{thm:non_markovian_bound} in the main text, i.e.,
\begin{equation}
|\sigma_{x'y'}| \leq \sum_{x,y}^{(\omega_{x'y'})} \sqrt{P_{x'|x} P_{y'|y}}|\rho_{xy}|.
\end{equation} 

\begin{ex}
	\label{ex:qubitchoi}
	In the qubit case, introduced in Example~\ref{ex:qubit1}, the Choi state is block-diagonal with blocks spanned by $\{\ket{00},\ket{11}\}$, $\{\ket{10}\}$ and $\{\ket{01}\}$:
	\small
	\begin{equation*}
	J[\E]=\frac{1}{2}\left[
	{
		\renewcommand{\arraystretch}{1.5}
		\begin{array}{cccc}
		P_{0|0}&C^{0|0}_{1|1}&0&0\\
		\left(C^{0|0}_{1|1}\right)^*&P_{1|1}&0&0\\
		0&0&1-P_{0|0}&0\\
		0&0&0&1-P_{1|1}\\
		\end{array}
	}
	\right],
	\end{equation*}
	\normalsize
	as in the elementary example of Sec.~\ref{sec:elementaryexample}.
	Positivity thus reads
	\begin{equation*}
	|\sigma_{10}|\leq \sqrt{P_{1|1}P_{0|0}}|\rho_{10}|.
	\end{equation*}
\end{ex}
	
\section{Covariant Markovian maps}
\label{appendix:markovian}
\subsection{Definition and Holevo's characterisation}

According to \ref{ass:markov}, a Markovian evolution is given by a one-parameter family of quantum channels constituting a quantum dynamical semigroup. The general form of the generator $\L_H$ is given by \cite{lindblad1976generators,gorini1976completely}
\begin{equation}
\L_H(\cdot)=\A(\cdot)-\tfrac{1}{2}\{\A^{\dag}(\iden),\cdot\}-i[\cdot,H],
\end{equation}
where $H$ is a Hermitian operator, $\A$ is a CP map, $\A^\dag$ is the adjoint of $\A$ (with respect to the Hilbert-Schmidt inner product, \mbox{$\tr{}{\rho\A(\sigma)}=\tr{}{\A^{\dag}(\rho)\sigma}$}) and $\{\cdot,\cdot\}$ denotes the anticommutator.

If the channel $\mathcal{E}=e^{\mathcal{L}_H t}$ generated by $\mathcal{L}_H$ is covariant, then $\L_H$ must be covariant, and it has been shown that both $\A$ and $\mathcal{H}(\cdot) = [H,\cdot]$ can also be chosen to be covariant \cite{holevo1993note}. Let \mbox{$\L=\L_H - i\mathcal{H}$}. Since \mbox{$[\L,\mathcal{H}]=0$} the evolution induced by $\L_H$ and $\L$ is the same up to an energy-preserving unitary: \mbox{$\E =e^{i\mathcal{H}_t} \circ e^{\L t}$}, where $e^{i\mathcal{H}_t} = e^{-i H t}(\cdot)e^{i Ht}$. Hence, from now on, we will ignore the term involving $H$ and consider Lindbladians of the form
\begin{equation}
\label{eq:lindbladian2}
\L(\cdot)=\A(\cdot)-\tfrac{1}{2}\{\A^{\dag}(\iden),\cdot\}.
\end{equation}

Since $\L$ is covariant, it acts on each mode independently (see Sec.~\ref{app_sec:modesofcoherence}). The action of $\L$ on the diagonal of a density matrix is therefore completely described by a matrix $L$ with elements
\begin{equation}
\label{eq:ldef}
L_{x'|x}:=\bra{x'} \mathcal{L}(\ketbra{x}{x})\ket{x'}.
\end{equation}
Recall that the covariant map $\mathcal{E}$ acts on the diagonal elements as a stochastic matrix $P$, hence $L$ is the generator of such a  matrix. In other words, $P$ must be an embeddable stochastic matrix \cite{davies2010embeddable}: $P=e^{L t}$. This implies that $L$ can be interpreted as a matrix of population transfer rates, similarly to the elementary example of Sec.~\ref{sec:elementaryexample}. Hence, $L$ satisfies $L_{x'|x}\geq 0$ for $x'\neq x$ and $\sum_{x'} L_{x'|x}=0$ \cite{davies2007linear}, which implies that $L_{x|x}\leq 0$ for all $x$. 

\subsection{Conditions on the generators of covariant Markovian maps}

Recall that $\A$, which appears in Eq.~\eqref{eq:lindbladian2}, is a covariant CP (but not necessarily TP) map. Denote the matrix elements of $\A$ by
\begin{equation}
A^{x'|x}_{y'|y} = \bra{x'}\A(\ketbra{x}{y})\ket{y'}.
\end{equation}
Reasoning as in Sec.~\ref{app_sec:positivitychoi}, the complete positivity of $\mathcal{A}$ is enforced by requiring $J[\mathcal{A}]\geq0$. Then the set of conditions
\begin{equation}
\label{eq:a_CP}
|A^{x'|x}_{y'|y}| \leq \sqrt{A_{x'|x} A_{y'|y}},
\end{equation}
where \mbox{$A_{x'|x}:=A^{x'|x}_{x'|x}$}, is necessary for ensuring that $\A$ is CP.

We now state some useful relations between the matrix elements of $\mathcal{A}$ and $\mathcal{L}$. Using the covariance of $\A$ it is straightforward to show that \mbox{$\mathcal{\A^\dag}(\iden) = \sum_{x',x} A_{x'|x} \ketbra{x}{x}$}. Hence we obtain
\begin{equation}
\label{eq:anticomm_els}
\bra{x'}\{\A^{\dag}(\iden),\ketbra{x}{y}\}\ket{y'} = \sum_z (A_{z|x}+A_{z|y})\delta_{xx'}\delta_{yy'},
\end{equation}
where $\delta_{xx'}$ denotes the Kronecker delta. So, in terms of the elements of $\A$ and $\L$, we have \mbox{$L_{x'|x} = A_{x'|x} - \sum_z A_{z|x} \delta_{x x'}$}. Importantly, \textit{any} element $L_{x'|x}$ can be expressed purely in terms of elements $A_{x'|x}$ for which $x'\neq x$:
\begin{equation}
\label{eq:a_l_equiv}
L_{x'|x} =
\begin{cases}  
-{\displaystyle \sum_{z\neq x} A_{z|x}} & \mbox{if } x'=x, \\
A_{x'|x} &\mbox{if } x'\neq x.
\end{cases}
\end{equation}

\section{Proof of Theorem 1}
\label{appendix:proof}

\markovianbound*

\noindent We first recall some notation and present identities that will be used in the proof. As in the main text, we define damping rates \mbox{$\gamma_{x'y'}:=(|L_{x'|x'}| + |L_{y'|y'}|)/2$} and express matrix elements in terms of their magnitudes and phase factors as $\rho_{xy}=|\rho_{xy}|\vartheta_{xy}$. The following two identities for the evolution of these terms are readily derived:
		\begin{align}
		\frac{d|\rho_{xy}|}{dt}&=\Re\left(\vartheta^*_{xy}\frac{d\rho_{xy}}{dt}\right),\label{eq:evolve_mod}\\
		|\rho_{xy}|\frac{d\vartheta_{xy}}{dt}&=\frac{d\rho_{xy}}{dt}-\vartheta_{xy}\frac{d|\rho_{xy}|}{dt} \label{eq:evolve_phase}.
		\end{align}
The strategy for the proof is as follows:
\begin{enumerate}[label=(\roman*)]
	\item Express the evolution of the absolute values of the density matrix element $|\rho_{x'y'}|$ in terms of the matrix elements of $\A$.
	
	\item Show that this expression may be bounded as
	\small
		\begin{equation}
		\label{eq:boundrateinequality}
		\frac{d|\rho_{x'y'}|}{dt} \leq -\gamma_{x'y'} |\rho_{x'y'}|  + \sum_{\mathclap{\substack{x\neq x'\\y \neq y'}}}^{\ompp{x}{y}} \sqrt{L_{x'|x} L_{y'|y}} |\rho_{xy}|.
		\end{equation}
	\normalsize
		
	\item Show that the solution of Eq.~\eqref{eq:boundrate}, i.e., of Eq.~\eqref{eq:boundrateinequality} with inequality sign replaced by an equality, gives an upper bound for the maximum coherence preservation. In other words, Eq.~\eqref{eq:bound} holds, and all that remains is to prove the tightness claims.
	
	\item Make a particular choice of $A^{x'|x}_{x'|x}$ (whilst leaving $L_{x'|x}$ unchanged), and find the evolution of the phase factor $\vartheta_{x'y'}$ under this choice.
	\item Demonstrate that this choice leaves $\vartheta_{x'y'}$ invariant when the phase-matching condition is satisfied.
	\item Verify that when the initial phase-matching condition holds, there is indeed a covariant CP map $\A$ that achieves the bound of Eq.~\eqref{eq:boundrateinequality}.

\end{enumerate}

\begin{proof}
	\begin{enumerate}[leftmargin=0cm,itemindent=0.7cm,labelwidth=\itemindent,labelsep=0cm,align=left,label=(\roman*)]
		\item The evolution of the system is given by \mbox{$d\rho/dt=\L\rho$}. Since $\L$ is covariant, each mode evolves independently as [see Eq.~\eqref{eq:covariancemodes}]
		\begin{equation}
		\label{eq:drho_dt_all}
		\frac{d\rho^{(\omega)}}{dt}=\L\rho^{(\omega)} = \sum^{(\omega)}_{x,y} \mathcal{L}(\ketbra{x}{y}) \rho_{xy}.
		\end{equation}
		According to Eq.~\eqref{eq:lindbladian2} and Eq.~\eqref{eq:anticomm_els} the evolution of any element $\rho_{x'y'}$ is
		\begin{equation}
		\label{eq:ratewithphase}
		\resizebox{.85\hsize}{!}{$\displaystyle{
				\frac{d\rho_{x'y'}}{dt}=\sum^{\ompp{x}{y}}_{x,y} \left[A^{x'|x}_{y'|y} - \frac{1}{2}\sum_z(A_{z|x} + A_{z|y})\delta_{xx'}\delta_{yy'}\right]\rho_{xy}.}$}
		\end{equation}		
	
We then use the identity given in Eq.~\eqref{eq:evolve_mod} to find an expression for $d|\rho_{x'y'}|/dt$. The summation over $x,y$ may be split up to isolate the term $A^{x'|x'}_{y'|y'}$, leaving a sum over indices $x,y$ such that \mbox{$(x,y)\neq(x',y')$}. Owing to covariance and non-degeneracy, this is equivalent to a sum such that $x \neq x'$ and $y\neq y'$ independently. We thus obtain
		\begin{equation}
		\label{eq:drho_dt_exact}
		\resizebox{.85\hsize}{!}{$\displaystyle{
				\frac{d|\rho_{x'y'}|}{dt}=-\Re(\alpha_{x' y'})|\rho_{x' y'}| + \sum_{\mathclap{\substack{x\neq x'\\y \neq y'}}}^{\ompp{x}{y}} \Re\left(A^{x'|x}_{y'|y} \vartheta^*_{x'y'} \rho_{xy} \right),}$}
		\end{equation}
		where \mbox{$\alpha_{x'y'} = \frac{1}{2}\sum_z(A_{z|x'} + A_{z|y'}) - A^{x'|x'}_{y'|y'}$}. Note that in the case of a non-degenerate Bohr spectrum this expression takes the particularly simple form
		\begin{equation}
		\label{eq:drho_dt_exact_nondegenerate}
		\frac{d |\rho_{x'y'}|}{dt} = - \Re(\alpha_{x'y'})|\rho_{x'y'}|.
		\end{equation}
		
		\item We now place bounds on the terms in this expression. Since $\A$ is CP we have [see Eq.~\eqref{eq:a_CP}]
		\begin{equation}
		\resizebox{.85\hsize}{!}{$\displaystyle{
				\Re\left(A^{x'|x}_{y'|y} \vartheta^*_{x'y'} \rho_{xy} \right)\leq |A^{x'|x}_{y'|y}| |\rho_{xy}| \leq \sqrt{A_{x'|x} A_{y'|y}} |\rho_{xy}|}$}
		\end{equation}
		and 
		\small
		\begin{align}
		\label{eq:rebound}
		\Re(\alpha_{x'y'}) &\geq \frac{1}{2}\sum_z \left( A_{z|x'} + A_{z|y'}\right)- |A^{x'|x'}_{y'|y'}| \nonumber \\
		&\geq \frac{1}{2}\sum_z \left( A_{z|x'} + A_{z|y'}\right)-\sqrt{A_{x'|x'} A_{y'|y'}} \nonumber  \\
		&\geq \frac{1}{2}\left( \sum_{z \neq x'} A_{z|x'} + \sum_{z \neq y'} A_{z|y'}\right) = \gamma_{x'y'},
		\end{align}
		\normalsize
		where the final inequality follows from the arithmetic-geometric mean inequality, and to get the final equality we use Eq.~\eqref{eq:a_l_equiv}. Applying the above bounds to Eq.~\eqref{eq:drho_dt_exact} and translating the expression into elements of $L_{x'|x}$ using Eq.~\eqref{eq:a_l_equiv} gives the bound on evolution that we seek, as given in Eq.~\eqref{eq:boundrateinequality}.
		
		\item Collecting all elements $|\rho_{x'y'}|$ of a given mode $\omega$ into a coherence vector $\v{c}^{(\omega)}$, and analogously for $\widetilde{\rho}_{x'y'}$ with a corresponding vector $\tilde{\v{c}}^{(\omega)}$, Eq.~\eqref{eq:boundrate} and Eq.~\eqref{eq:boundrateinequality} read
		\begin{equation}
		\label{eq:c}
		\frac{d\tilde{\v{c}}^{(\omega)}}{dt} = Q \tilde{\v{c}}^{(\omega)}, \quad \frac{d\v{c}^{(\omega)}}{dt} \leq Q \v{c}^{(\omega)},
		\end{equation}
		with initial conditions $\v{c}^{(\omega)}(0)= \tilde{\v{c}}^{(\omega)}(0)$ and the vector inequality denoting a set of component-wise inequalities. The off-diagonal elements of $Q$ are given by $\sqrt{L_{x'|x}L_{y'|y}}$ and are hence non-negative. Thus, Lemma~\ref{lem:differentialinequality} from Appendix~\ref{appendix:systemsofodes} implies that $\tilde{\v{c}}^{(\omega)}(t) \geq \v{c}^{(\omega)}(t)$ for all $t\geq 0$. This can be rewritten as the bound of Eq.~\eqref{eq:bound} for all elements of the mode. The same reasoning applies to any mode $\omega$, and so the result holds for every off-diagonal element of~$\rho$.
		
		\item We now begin our proof of attainability of the bound in Eq.~\eqref{eq:boundrateinequality}. Consider setting $A_{x'|x'}=0$ for all $x'$. The necessary condition for $\A$ to be CP, Eq.~\eqref{eq:a_CP}, then also imposes $A_{y'|y'}^{x'|x'}=0$. Note that this choice does not alter any element $L_{x'|x}$, which can be expressed using only elements $A_{x'|x}$ for which $x\neq x'$. Using Eq.~\eqref{eq:a_l_equiv} we find that $\Re(\alpha_{x'y'})=\gamma_{x'y'}$ and hence Eq.~\eqref{eq:drho_dt_exact} becomes
		\begin{equation}
		\label{eq:drho_dt_exact_pm}
		\resizebox{.85\hsize}{!}{$\displaystyle{
				\frac{d|\rho_{x'y'}|}{dt}=-\gamma_{x'y'}|\rho_{x' y'}| + \sum_{\mathclap{\substack{x\neq x'\\y \neq y'}}}^{\ompp{x}{y}} \Re\left(A^{x'|x}_{y'|y} \vartheta^*_{x'y'} \rho_{xy} \right).}$}
		\end{equation}
		Note that in the case of a non-degenerate Bohr spectrum the above equation gives \mbox{$d|\rho_{x'y'}|/dt=-\gamma_{x'y'} |\rho_{x'y'}|$} for all $x',y'$, which leads to simultaneous saturation of the bound for all coherence elements.			
			
		For the more complicated general case, from Eq.~\eqref{eq:ratewithphase} and with our particular choice of $A_{x'|x'}$ we have
		\begin{equation}
		\label{eq:drho_dt_all_pm}
		\frac{d\rho_{x'y'}}{dt}=-\gamma_{x'y'}\rho_{x' y'} + \sum_{\mathclap{\substack{x\neq x'\\y \neq y'}}}^{\ompp{x}{y}} A^{x'|x}_{y'|y} \rho_{xy}.
		\end{equation}
		The identity provided in Eq.~\eqref{eq:evolve_phase} then gives the evolution of the phase factor:
		\begin{equation}
		\resizebox{.85\hsize}{!}{$\displaystyle{
				|\rho_{x'y'}|\frac{d \vartheta_{x'y'}}{dt} = \sum_{\mathclap{\substack{x\neq x'\\y \neq y'}}}^{\ompp{x}{y}} \left[A^{x'|x}_{y'|y} \rho_{xy} - \vartheta_{x'y'}  \Re\left(A^{x'|x}_{y'|y} \vartheta^*_{x'y'} \rho_{xy}\right) \right]}.$}
		\end{equation}
		Since \mbox{$\rho_{xy}=|\rho_{xy}|\vartheta_{xy}$} and \mbox{$\rho^*_{xy}=|\rho_{xy}|\vartheta^*_{xy}$}, we find that
		\begin{equation}
		\label{eq:phase_evolve}
		\resizebox{.85\hsize}{!}{$\displaystyle{
				\frac{d \vartheta_{x'y'}}{dt} \propto \sum_{\mathclap{\substack{x\neq x'\\y \neq y'}}}^{\ompp{x}{y}} \left(A^{x'|x}_{y'|y} \vartheta_{xy} \vartheta^*_{x'y'} - {A^{x'|x}_{y'|y}}^* \vartheta_{x'y'} \vartheta^*_{xy} \right) |\rho_{xy}|}.$}
		\end{equation}
		
		\item Consider now the process that, for every $\ketbra{x}{y}$ belonging to the same mode as $\ketbra{x'}{y'}$, satisfies
		\begin{equation}
		\label{eq:phase_a}
		\frac{A^{x'|x}_{y'|y}}{|A^{x'|x}_{y'|y}|} = \vartheta_{x'y'}(0) \vartheta^*_{xy}(0).
		\end{equation}
		Note that again this choice does not affect $L_{x'|x}$. Eq.~\eqref{eq:phase_evolve} then becomes
		\small
		\begin{align}
		\frac{d \vartheta_{x'y'}}{dt} \propto \sum_{\mathclap{\substack{x\neq x'\\y \neq y'}}}^{\ompp{x}{y}} (&\vartheta_{x'y'}(0) \vartheta^*_{xy}(0) \vartheta_{xy} \vartheta^*_{x'y'} \nonumber \\
		-&\vartheta_{xy}(0) \vartheta^*_{x'y'}(0) \vartheta_{x'y'} \vartheta^*_{xy} ) |A^{x'|x}_{y'|y}| |\rho_{xy}|,
		\end{align}
		\normalsize
		which may be solved by taking phases factors constant for all $t$: \mbox{$\vartheta_{x y}=\vartheta_{x y}(0)$} and \mbox{$\vartheta_{x' y'}=\vartheta_{x'y'}(0)$}. Using this solution, Eq.~\eqref{eq:drho_dt_exact_pm} becomes
		\begin{equation}
		\frac{d|\rho_{x'y'}|}{dt}=-\gamma_{x'y'}|\rho_{x' y'}| + \sum_{\mathclap{\substack{x\neq x'\\y \neq y'}}}^{\ompp{x}{y}} |A^{x'|x}_{y'|y}||\rho_{xy}|.
		\end{equation}
		We can now choose $\A$ to be at the boundary of the CP constraint by fixing
		\begin{equation}
		\label{eq:a_CP_bound}
		|A^{x'|x}_{y'|y}| = \sqrt{A_{x'|x}A_{y'|y}}.
		\end{equation}
		Recalling that $A_{x'|x} = L_{x'|x}$ for $x \neq x'$, we conclude that under the above choices the inequality \eqref{eq:boundrateinequality} is tight.
		
		\item Finally, we show that there is indeed a covariant CP map $\A$ that realises the above evolution when the initial phase-matching condition holds. Consider a quantum channel $\A$ given by its Kraus decomposition $\{K_\omega\}$ with
		\begin{equation}
		\label{eq:krausoptimal}
		K_{\omega} = \sum^{(\omega)}_{x,y} \vartheta_{xy}(0) \sqrt{A_{x|y}}\ketbra{x}{y}.
		\end{equation}
		Using \mbox{$A^{x'|x}_{y'|y} = \sum_{\omega}\bra{x'}K_\omega \ketbra{x}{y} K^\dag_\omega \ket{y'}$} we find that
		\begin{equation}
		\label{eq:optimal_A}
		A^{x'|x}_{y'|y} = \vartheta_{x'x}(0) \vartheta^*_{y'y}(0)  \sqrt{A_{x'|x}A_{y'|y}}.
		\end{equation}
		Applying the phase-matching condition gives
		\begin{equation}
		\label{eq:finalchoiceofA}
		A^{x'|x}_{y'|y} = \vartheta_{x'y'}(0) \vartheta^*_{xy}(0)  \sqrt{A_{x'|x}A_{y'|y}},
		\end{equation}
		which is readily seen to satisfy Eqs.~\eqref{eq:phase_a} and \eqref{eq:a_CP_bound}, as well as allowing the choice $A_{x'|x'}=0$ that results in $A^{x'|x'}_{y'|y'}=0$. The covariance of $\A$ immediately follows from Proposition~7 of Ref.~\cite{marvian2016quantify}, which states that if each $K_\omega$ in the Kraus decomposition is a mode~$\omega$ operator [in the same sense as each $\rho^{(\omega)}$ in the decomposition of Eq.~\eqref{eq:modesofcoherencerho}], then the induced map is covariant. 
		
Making the choice given in Eq.~\eqref{eq:finalchoiceofA} for every element $x',y'$ of mode $\omega$, we see that our inequality for $d\v{c}^{(\omega)}/dt $ [Eq.~\eqref{eq:c}] is saturated. Hence, with these choices, the evolution of every element in a mode achieves the claimed bound tightly.

	\end{enumerate}
	
\end{proof}

\section{Proof of tightness of Theorem \ref{thm:non_markovian_bound}}
\label{app:proof_of_tightness}

In Appendix~\ref{app_sec:positivitychoi} we proved that the bound specified in Eq.~\eqref{eq:non_markovian_bound} of Theorem~\ref{thm:non_markovian_bound} holds, a result first given in Ref.~\cite{lostaglio2015quantum}. Here we show under what conditions the bound is tight. We achieve this by providing an explicit construction of a covariant channel that saturates the bound.

Recall that we require the Choi state $J[\E]$ to be positive semidefinite. By noting that
\begin{equation}
	\sum_{x',x} \sum^{\omp{x}{x}}_{y',y} (\cdot)=	\sum_{\omega}\sum_{x',x}^{(\omega)} \sum^{(\omega)}_{y',y} (\cdot)
\end{equation}
we can rewrite Eq.~\eqref{eq:choi} in the following form, which makes the block-diagonal structure of $J[\E]$ evident:
\begin{equation}
	\label{eq:choi2}
	J[\E]= \frac{1}{d} \sum_{\omega}\sum_{x',x}^{(\omega)} \sum^{(\omega)}_{y',y} C^{x'|x}_{y'|y}\ketbra{x'x}{y'y},
\end{equation}

Now, given any population transfer matrix $P$, we can choose each block $\omega$ of the Choi state to be an unnormalised pure state $\ketbra{\psi_{\omega}}{\psi_{\omega}}$, where 
\begin{equation}
\ket{\psi_{\omega}}=\frac{1}{\sqrt{d}}\sum^{(\omega)}_{x',x} \varphi_{x'x} \sqrt{P_{x'|x}} \ket{x'x}
\end{equation}
and $\varphi_{x'x}$ are phase factors. This way we ensure positivity of $J[\E]$ and the corresponding quantum channel is given by \mbox{$\E(\cdot) = \sum_{\omega}K_{\omega} (\cdot) K_{\omega}^\dag$} with Kraus operators
\begin{equation}
\label{eq:optimal_kraus}
K_{\omega}=\sum^{(\omega)}_{x',x} \varphi_{x'x} \sqrt{P_{x'|x}} \ketbra{x'}{x}.
\end{equation}
Using Proposition~7 of Ref.~\cite{marvian2016quantify} one can directly verify that these Kraus operators generate a time-translation symmetric channel. The matrix elements $C^{x'|x}_{y'|y} $ are given by
\begin{equation}
\label{eq:optimal_choi}
C^{x'|x}_{y'|y} = \varphi_{x'x}\varphi^*_{y'y} \sqrt{P_{x'|x}P_{y'|y}}.
\end{equation}
Such a channel transforms populations according to the population transfer matrix $P$; and writing $\sigma=\E(\rho)$, we find that coherence terms transform as
\begin{equation}
\label{eq:choi_coherence}
\sigma_{x'y'} = \sum_{x,y}^{\ompp{x}{y}} \sqrt{P_{x'|x} P_{y'|y}}|\rho_{xy}|\varphi_{x'x}\varphi_{y'y}^*\vartheta_{xy},
\end{equation}
where we recall that \mbox{$\rho_{xy}=|\rho_{xy}|\vartheta_{xy}$}.

Now the crucial question is whether the phase factors $\{\varphi_{x'x}\}$ can be chosen in such a way as to saturate the bound given in Eq.~\eqref{eq:non_markovian_bound}. Comparing Eqs.~\eqref{eq:choi_coherence}~and~\eqref{eq:non_markovian_bound}, we see that the choice \mbox{$\varphi_{x'x} \varphi^*_{y'y}= \vartheta^*_{xy}$} ensures saturation of the bound. First, let us consider a simple case when the mode $\omega$ contains no overlapping elements, i.e., for every two distinct coherence terms $\sigma_{x'y'}$ and $\sigma_{xy}$ in mode $\omega$, all indices $x,y,x',y'$ are distinct. Then, for every $x',y'$ we can make the choice of phases \mbox{$\varphi_{x'x} = \vartheta^*_{xy}$} and $\varphi^*_{y'y}=1$ for all $x,y$ such that $\omega_{xy} = \omega_{x'y'} =\omega$ (note that \mbox{$\varphi_{x'x} = \vartheta^*_{xy}$} is meaningful as, due to non-degeneracy of the Hamiltonian, a single index $x$ is enough to specify $y$). This leads to saturation of the bound. 

In the general case a mode~$\omega$ consists of coherence elements $\sigma_{x_n x_{n-1}}, \sigma_{x_{n-2} x_{n-3}} \dots, \sigma_{x_{1} x_0}$ with $x_{i}$ sorted in non-decreasing energy order (\mbox{$\omega_{x_{i}} \leq \omega_{x_{i+1}}$}), meaning that some of $x_i$ are equal to $x_{i+1}$ (corresponding to overlapping coherence elements). One can now make the following choice of $\{\varphi_{x'x}\}$ to saturate the bound given in Eq.~\eqref{eq:non_markovian_bound}. We set initial conditions \mbox{$\varphi_{x_i x_0} = \varphi_{x_0 x_i} = 1$} for all $i=0,...,n$, and impose iteratively \mbox{$\varphi_{x_{i+1}x_{j+1}}=\varphi_{x_i x_{j}}\vartheta^*_{x_{j+1} x_{j}}$} for all $i,j=0,...,n-1$. This choice is depicted in Fig.~\ref{fig:phase_match_NM} for the example case of a 3-element mode.
\begin{figure}
	\includegraphics[width=0.8\columnwidth]{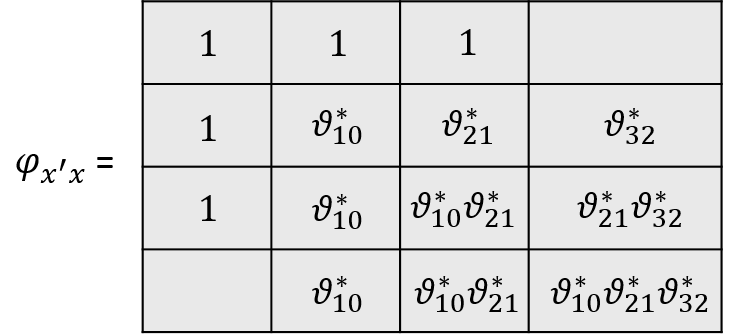}
	\caption{\label{fig:phase_match_NM} \emph{Choice of phases for overlapping elements of a mode.} The choice of phases $\varphi_{x' x}$ that saturates the bound of Theorem~\ref{thm:non_markovian_bound} for all elements of a mode. The phases are given for the case of three overlapping elements: $\sigma_{10}$, $\sigma_{21}$ and $\sigma_{32}$ (hence \mbox{$x_i = 0,...,3$}), but the table can be easily extended noticing the structure of each diagonal. The phases in the corners of the matrix, in this case $\varphi_{30}$ and $\varphi_{03}$, do not need to be fixed to saturate the bound.}
\end{figure}

Moreover, if it happens that phase factors of the initial state $\rho$ are of the form $\vartheta_{xy}= \phi_x \phi^*_y$, one can saturate the bound simultaneously for all coherence elements. This can be achieved by choosing $\varphi_{x'x}= \phi_x^*$ for all $x,x'$. Finally, if the Bohr spectrum is non-degenerate, the summation in Eq.~\eqref{eq:choi_coherence} consists of a single term with $x=x'$ and $y=y'$. Hence, the bound is saturated independently of the choice of $\varphi_{x'x}$.

\section{Proof of Corollary~\ref{corol:attractor2}}
\label{appendix:cor4}

\attractor*

\begin{proof}

We begin by proving that uniqueness of the fixed point implies that all transition probabilities, $P_{x'|x}(t)$, are non-zero at some finite time $\tilde t>0$ (recall that \mbox{$P(t) = e^{L t}$}). As we have seen in Eq.~\eqref{eq:pop_evol2}, if the fixed point is unique then every initial distribution converges to $\v{\pi}$ as $t \rightarrow \infty$. Consider the set of standard basis vectors $\{\v{\epsilon}^x\;|\;x=1,\dots,d\}$, where $\v{\epsilon}^x$ denotes the unit vector with a $1$ in the $x^\mathrm{th}$ coordinate and $0$s elsewhere. We have $P(t)\v{\epsilon}^x \rightarrow \v{\pi}$ as $t \rightarrow \infty$. Hence, for every $x,x'$ there exists some $t_{x,x'}>0$ such that $P_{x'|x}(t)>0$ for all $t \geq t_{x,x'}$. Taking \mbox{$\tilde t = \max_{x,x'}t_{x,x'}$}, we have \mbox{$P_{x'|x}(\tilde t) > 0$} for every $x,x'$. 	
	
We now prove the second claim that all coherence terms must vanish as $t\rightarrow\infty$. Consider the elements of a mode~$\omega$. From Theorem~\ref{thm:non_markovian_bound} we have
\begin{equation}
|\rho_{x'y'}(\tilde t)| \leq \sum_{x,y}^{(\omega)} \sqrt{P_{x'|x}(\tilde t) P_{y'|y}(\tilde t)}|\rho_{xy}(0)|.
\end{equation}
The definition of $S_\omega$ given in Eq.~\eqref{eq:coherenceonenorm} yields
	\begin{eqnarray*}
		S_\omega(\rho(\tilde t)) &\leq& \sum_{x,y}^{(\omega)} |\rho_{xy}(0)|  \sum_{x',y'}^{(\omega)} \sqrt{P_{x'|x}(\tilde t) P_{y'|y}(\tilde t)} \\
		& \leq &  \sum_{x,y}^{(\omega)} |\rho_{xy}(0)| B_{xy}(\tilde t),
	\end{eqnarray*}
	where we have used the arithmetic-geometric mean inequality and defined
	\begin{equation}
	B_{xy}(\tilde t):=\tfrac{1}{2}\sum_{x',y'}^{(\omega)}\left(P_{x'|x}(\tilde t) + P_{y'|y}(\tilde t)\right).
	\end{equation}
There are at most $d-1$ terms in this summation. Given that \mbox{$\sum^d_{x'=1}P_{x'|x}(\tilde t) =1$} and that \mbox{$P_{x'|x}(\tilde t)>0$} for every $x',x$, it follows that \mbox{$\sum_{x',y'}^{(\omega)}P_{x'|x}(\tilde t) < 1$}, so that \mbox{$B_{xy}(\tilde t)<1$} for every $x,y$. Taking \mbox{$B = \max_{x,y}B_{xy}(\tilde t)$}, we have \mbox{$S_\omega(\rho(\tilde t)) \leq B 	S_\omega(\rho(0))$} with $B<1$. 
	
As the process is time-homogeneous, the presented reasoning leads to
	\begin{equation}
	S_\omega(\rho((n+1)\tilde t)) \leq B S_\omega(\rho(n\tilde t)), 
	\end{equation}
	for all $n \in \mathbb{N}$. This implies \mbox{$S_\omega(\rho(n\tilde t))< B^n S_\omega(\rho(0))$}. Taking the limit $n \rightarrow \infty$ we obtain \mbox{$ S_\omega(\rho(t)) \rightarrow 0$} as \mbox{$t \rightarrow \infty$}, which implies that \mbox{$\lim_{t\rightarrow\infty}|\rho_{x'y'}(t)| = 0$} for every $x,y$ with $\omega_{x'}-\omega_{y'} = \omega$. The same reasoning applies to every mode of coherence $\omega$, and hence proves the second claim.
	
Let us now prove the first claim. To achieve this we will use a theorem of L\'{e}vy~\cite{freedman1983approximating} that states that either $P_{x'|x}(t)=0$ for all $t\geq 0$ or $P_{x'|x}(t)>0$ for all \mbox{$t>0$}. Since we proved that $P_{x'|x}(\tilde t)>0$ for all $x',x$, we conclude that $P_{x'|x}(t)>0$ for all $t>0$ and for all $x',x$. Using this result we can apply the same reasoning as in the discussion above to show that for every $t'>t$ we have $S_\omega(\rho(t')) \leq B S_\omega(\rho(t))$, with $B<1$, and hence $S_\omega(\rho(t')) < S_\omega(\rho(t))$. 
	
\end{proof}

\section{Optimal coherence transfer for a qutrit}
\label{appendix:qutrit}

As stated in the main text, for a qutrit with equidistant spectrum, optimal coherence evolution is governed by $d\v c/dt=Q\v c$, where
\small
\begin{equation*}\renewcommand{\arraystretch}{1.3}
Q=\left[
\begin{array}{cc}
-\gamma_{10} & \sqrt{L_{0|1}L_{1|2}} \\
\sqrt{L_{1|0}L_{2|1}} & -\gamma_{21}
\end{array}
\right],
\end{equation*}	
\normalsize
and \mbox{$\v{c}:=(|\rho_{10}|, |\rho_{21}|)$}. We recall the definition of the damping rate, \mbox{$\gamma_{xy}=(|L_{x|x}| + |L_{y|y}|)/2$}. The (magnitude of) the coherence element $\rho_{21}(t)$ evolves according to Eq.~\eqref{eq:oscillator_equation}, which we repeat here for convenience:
\begin{equation}
\label{eq:c_2evol}
\frac{d^2 c_2}{dt^2}-\mathrm{Tr}(Q) \frac{d c_2}{dt}+\det(Q) c_2=0.
\end{equation}
We wish to achieve optimal transfer of coherence from density matrix element $\rho_{10}$ to density matrix element $\rho_{21}$. Concretely, we find $\max_{L,t}c_2(t)$, where the optimisation runs over all population transfer rates $L_{x'|x}$ and times $t$.

The dynamics described by Eq.~\eqref{eq:c_2evol} are precisely those of a damped harmonic oscillator with damping $\eta$ and natural frequency $\nu$, where:
\begin{align}
\eta&=-\mathrm{Tr}(Q)=\gamma_{10}+\gamma_{21},\\
\nu^2&=\det(Q)=\gamma_{10}\gamma_{21}-\sqrt{L_{0|1}L_{1|2}L_{1|0}L_{2|1}}.
\end{align}
Given these expressions, we define \mbox{$D:=\eta^2-4\nu^2$}, which can be straightforwardly evaluated as
\begin{equation}
D=(\gamma_{10}-\gamma_{21})^2+4\sqrt{L_{0|1}L_{1|2}L_{1|0}L_{2|1}}.
\end{equation}
It is clear that $D\geq 0$, which corresponds to an overdamped ($D>0$) or critically damped ($D=0$) oscillator.

For the case that $D>0$, the solution of Eq.~\eqref{eq:c_2evol} is given by \mbox{$c_2(t)=A_+e^{p_+t}+A_-e^{p_-t}$}, where \mbox{$p_\pm=\frac12(-\eta\pm\sqrt{D})$} and $A_\pm$ are constants determined by the initial conditions. Given that $c_2(0)=0$ and $\frac{dc_2}{dt}(0)=\sqrt{L_{1|0}L_{2|1}}c_1(0)$, one obtains 
\begin{equation}
A_+=-A_-=\sqrt{L_{1|0}L_{2|1}}c_1(0)/\sqrt{D}.
\end{equation} 
The solution may then be written as
\begin{equation}
\label{eq:c_2full_soln}
c_2(t)=2c_1(0)\sqrt{\frac{L_{1|0}L_{2|1}}{D}}\,e^{-\tfrac12\eta t}\mathrm{sinh}\left(\tfrac12\sqrt{D} t\right).
\end{equation}
In addition to the prefactor involving $L_{1|0}$ and $L_{2|1}$, the evolution $c_2(t)$ depends on the matrix elements $L_{x'|x}$ through the expressions for $\eta$ and $D$. To analytically perform the full optimisation $\max_{L,t}c_2(t)$, subject to the constraints $L_{x|x}\leq 0$ and $\sum_{x'} L_{x'|x} = 0$, appears to be highly non-trivial. However, we can straightforwardly derive an upper bound for the optimal coherence transfer. Consider the solution $\widetilde c_2(t)$ that holds when $\nu=0$, so that $\sqrt{D}=\eta$. Eq.~\eqref{eq:c_2full_soln} then becomes
\begin{align}
\nonumber
\widetilde c_2(t)&=2c_1(0)\sqrt{\frac{L_{1|0}L_{2|1}}{\eta^2}}\,e^{-\tfrac12\eta t}\mathrm{sinh}\left(\tfrac12\eta t\right),\\
\label{eq:c_2_wt_soln2}
&=c_1(0)\sqrt{\frac{L_{1|0}L_{2|1}}{\eta^2}}\,(1-e^{-\eta t}).
\end{align}
The solution for $\widetilde c_2$ provides an upper bound for $c_2$, i.e., for all $t$ we have $c_2(t) \leq \widetilde c_2(t)$, and hence \mbox{$\max_t c_2(t)\leq \max_t \widetilde c_2(t)$}. To see this, we observe that
\begin{equation}
\frac{c_2}{\widetilde c_2}=\frac{\mathrm{sinhc}(\frac12 \sqrt{D} t)}{\mathrm{sinhc}(\frac12 \eta t)},
\end{equation}
where $\mathrm{sinhc}$ is the hyperbolic $\mathrm{sinc}$ function, \mbox{$\mathrm{sinhc}(z)=\mathrm{sinh}(z)/z$}. Since $\mathrm{sinhc}(z)$ is monotonically increasing for $z\geq 0$, and $\sqrt{D}\leq\eta$, we have $c_2(t)\leq \widetilde c_2(t)$. This inequality may also be thought of physically in terms of the analogy with a damped harmonic oscillator: $\nu$ gives a measure of the linear ``spring'' force. The displacement of an oscillator from its equilibrium position is always bounded by the displacement when there is no restoring force.

We now proceed with the optimisation $\max_{L,t} \widetilde c_2(t)$. Eq.~\eqref{eq:c_2_wt_soln2} achieves its maximum as $t\rightarrow \infty$, so that $\max_{t} \widetilde c_2(t)=c_1(0) f(L)$, where
\begin{equation}
f(L)=\frac{2\sqrt{L_{1|0}L_{2|1}}}{|L_{0|0}|+|L_{2|2}|+2|L_{1|1}|}.
\end{equation}
To perform the maximisation over $L$, we first note that the only dependence of $f(L)$ on elements $L_{x'|2}$ is through the $|L_{2|2}|$ in the denominator. To maximise $f(L)$ we may thus freely set $L_{2|2}=0$ (and hence the population transfer rate matrix will also have $L_{1|2}=L_{0|2}=0$). Given the constraints $\sum_{x'}L_{x'|0}=\sum_{x'}L_{x'|1}=0$, it is clear that $f(L)$ is maximised when \mbox{$L_{2|0}=L_{0|1}=0$}, so that $|L_{0|0}|=L_{1|0}$ and $|L_{1|1}|=L_{2|1}$. Hence we have \mbox{$f(L)=2\sqrt{L_{1|0}L_{2|1}}/(L_{1|0}+2L_{2|1})$}. According to the inequality of arithmetic and geometric means, this is maximised when we set $L_{1|0}=2L_{2|1}$, yielding \mbox{$\max_L f(L) = 1/\sqrt2$}. Hence, we find that 
\begin{equation}
\max_{L,t}c_2(t)\leq c_1(0)/\sqrt2\approx0.707 c_1(0),
\end{equation}
as given in the main text.

\section{Differential inequalities for a system of ODEs}
\label{appendix:systemsofodes}
Here we present a technical result that was used in the proof of Theorem~\ref{thm:markovian_bound}, but which may be of some independent interest. Results on differential inequalities tend to be limited to simple cases, e.g. Gr\"onwall's lemma~\cite{gronwall1919note} applies to the linear first order differential inequality $du(t)/dt\leq\alpha(t)u(t)$. We prove a result that can be applied to a system of linear first order differential inequalities, such as those governing the evolution of off-diagonal elements according to the Bloch equations. This result can likely be understood as a special case of general comparison theorems (see, e.g., Ref.~\cite{szarski1965differential}, Chap.~3). Here we give a proof that does not require such sophisticated technical machinery.

Given two $n$-dimensional vectors $\v{x}$ and $\v{y}$, the notation $\v{x} \geq \v{y}$ will denote the component-wise inequality \mbox{$x_i \geq y_i$} for all $i=1,\dots,n$.  A system of linear first order differential equations may be written as $\dot{\v{x}}(t)=M \v{x}(t)$, where $\dot{\v{x}}:=d\v{x}/dt$ and $M$ is some $n\times n$ matrix. If we instead have the differential inequality $\dot{\v{x}}(t)\leq M \v{x}(t)$ then what can be inferred about the evolution?

\begin{lem}
	\label{lem:differentialinequality}
	Let $\v{x}: \R \mapsto \R^n$ be the solution of the system of linear differential equations
	\begin{equation}
	\dot{\v{x}}(t) = M \v{x}(t), \quad \v{x}(0) = \v{u}, 
	\end{equation}
	where $M_{ij} \geq 0$ for all $i \neq j$. If $\v{y}: \R \mapsto \R^n$ is a differentiable function satisfying
	\begin{equation}
	\dot{\v{y}}(t) \leq M \v{y}(t), \quad \v{y}(0) = \v{u}, 
	\end{equation}
	then $\v{x}(t) \geq \v{y}(t)$ for all $t \geq 0$.
\end{lem}
\begin{proof}
Denote by $\P$ the set of differential equations and inequalities for $\v{x}$ and $\v{y}$ given in the statement of the Lemma. Given $\v{c} > \v{0}$, we first prove the following implication for a modified problem $\P'$:
\begin{equation}\label{eq:pprime}
\renewcommand{\arraystretch}{2}
\begin{rcases}
	\dot{\v{x}}(t) = M \v{x}(t) + \v{c} \\
	\dot{\v{y}}(t) < M \v{y}(t) + \v{c} \\
	\v{x}(0) > \v{y}(0)
	\end{rcases} \Longrightarrow \v{x}(t) \geq \v{y}(t) \; \forall t \geq 0.\\
\end{equation}

By continuity, there exists some $\tilde{t} >0$ such that \mbox{$\v{x}(t) > \v{y}(t)$} for all $t \in [0,\tilde{t})$. Let us define the `overtaking' time as
	\begin{equation}
	t^\star = \sup_{\tilde{t} > 0} \, \left\{\tilde{t}\;|\;\v{x}(\tilde{t})> \v{y}(\tilde{t}) \; \forall t \in [0,\tilde{t})\right\}.
	\end{equation}
By contradiction, assume $t^\star < \infty$. Continuity implies $\v{x}(t^\star) \geq \v{y}(t^\star)$. By definition of $t^\star$, we have that
	\begin{enumerate}
		\item There exists $i$ such that  $x_i(t^\star) = y_i(t^\star)$; \label{overcome1}
		\item There exists a sequence $\{t_\alpha\}$ with $t_\alpha \searrow t^\star$ such that \mbox{$x_i(t_\alpha)< y_i(t_\alpha)$} for all $\alpha$. \label{overcome2}
\end{enumerate}
Using condition~\ref{overcome1} and continuity:
\small
\begin{eqnarray*}
\dot{y}_i(t^\star) &<& \sum_j M_{ij}y_j(t^\star) + c_i = M_{ii}x_i(t^\star) + \sum_{j \neq i} M_{ij} y_j(t^\star) + c_i \\
		&\leq& M_{ii}x_i(t^\star) + \sum_{j \neq i} M_{ij} x_j(t^\star) + c_i = \dot{x}_i(t^\star).
\end{eqnarray*}\normalsize
But this implies that there is a right neighbourhood of $t^\star$ in which  $x_i(t)>y_i(t)$, which is in direct contradiction with condition~\ref{overcome2}. Hence it must be $t^\star = \infty$, i.e. overtaking can never take place and the implication given in Eq.~$\eqref{eq:pprime}$ holds: we have \mbox{$\v{x}(t) \geq \v{y}(t) \; \forall t \geq 0$}.

Now consider the following sequence of problems $\P_m$ with $m \in \mathbb{N}$:
\begin{equation*}
\left\{
\begin{array}{ll}
\dot{\v{x}}^{(m)}(t) = M \v{x}^{(m)}(t) + \v{c}^{(m)}, & \v{x}^{(m)}(0)= \v{u}^{(m)}, \\
\dot{\v{y}}(t) \leq M \v{y}(t),&\v{y}(0) = \v{u},
\end{array}\right.
\end{equation*}\vspace{0.01cm}

\noindent where $\v{c}^{(m)}:=(1/m,\dots,1/m)$ and $\v{u}^{(m)}\searrow \v{u}$, \mbox{$\v{u}^{(m+1)}< \v{u}^{(m)}$}. Clearly any $\P_m$ is a problem of the form $\P'$, since we have \mbox{$\dot{\v{y}}(t) < M\v{y}(t) + \v{c}^{(m)}$} and \mbox{$\v{x}^{(m)}(0)>\v{y}(0)$} (as $\v u^{(m)}>\v u$). Using Eq.~\eqref{eq:pprime} we therefore deduce that $\v{x}^{(m)}(t) \geq \v{y}(t)\; \forall t \geq 0, \forall m \in \mathbb{N}$.

Note that $\v{x}^{(\infty)}$ solves the desired problem $\P$ and so $\v{x}^{(m)} \rightarrow \v{x}$ pointwise as $m \rightarrow \infty$ (in fact, one can show the convergence is locally uniform, since comparing $\P_m$ and $\P_{m+1}$ gives \mbox{$\v{x}^{(m+1)}(t) \leq \v{x}^{(m)}(t)$}). Hence we have that \mbox{$\v{x}(t) \geq \v{y}(t)\;\forall t\geq 0$}, as required. 
\end{proof}

\section{Embeddable stochastic matrices}
	\label{appendix_embeddable}
	
	The proof of Theorem~\ref{thm:embeddable} requires the following results that we state without proofs, instead pointing to references where these can be found:
	
	\begin{lem}[Lemma~12.3.4~and~12.3.5 of Ref.~\cite{davies2007linear}]
		\label{thm:generator_form}
		A matrix $B$ is a generator of a stochastic matrix, i.e., \mbox{$A=e^B$} for some stochastic matrix $A$, if and only if it is of the form $B=\alpha(C-\iden)$ for some $\alpha\geq 0$ and stochastic matrix~$C$.
	\end{lem}
	
	\begin{lem}[Theorem~1.7 in Chap.~VII of Ref.~\cite{minc1988nonnegative}]
		\label{thm:part_karpelevich}
		A $d\times d$ stochastic matrix has no eigenvalues corresponding to points inside either of the two segments of the unit circle joining the point 1 with $e^{2\pi i/d}$ and $e^{-2\pi i/d}$. 
	\end{lem}
	
	\noindent We are now ready to present the proof of Theorem~\ref{thm:embeddable}.
	\begin{proof}
		By Lemma~\ref{thm:generator_form} every embeddable stochastic matrix $P$ is of the form $P=e^{\alpha(C-\iden)}$ with $\alpha\geq 0$ and $C$ a stochastic matrix. Hence, an eigenvalue $\lambda$ of $P$ is equal to $e^\mu$, where $\mu$ is an eigenvalue of $\alpha(C-\iden)$. Now, using Lemma~\ref{thm:part_karpelevich}, we have $|\!\arg\mu|\geq\pi/2+\pi/d$, so that \mbox{$\mu=-a+ib$} with $a\geq 0$ and \mbox{$|b|\leq a \tan^{-1}(\pi/d)$}.
		Introducing $r=e^{-a}$ and $\phi=b$ we obtain $\lambda=r e^{i\phi}$ and the constraint on $|b|$ is translated into
		\begin{equation*}
		e^{-|\phi| \tan \frac{\pi}{d}}\geq r\geq 0.
		\end{equation*} 
	\end{proof}

\section{Bounding transport rates}
\label{appendix:lbound}
	
From $L \v{\pi} = \v{0}$ and \mbox{$\pi_x \propto e^{-\beta \hbar \omega_x}$} one obtains for any fixed and distinct $x', x$
\begin{equation*}
L_{x'|x'}e^{-\beta \hbar \omega_{x'}}+L_{x'|x} e^{-\beta \hbar \omega_x} + \sum_{y \neq x',x} L_{x'|y} e^{-\beta \hbar \omega_y} = 0.
\end{equation*}
Recalling that $\omega_{x'x}:= \omega_{x'} - \omega_x$, the above can be rewritten as
\begin{equation*}
L_{x'|x'}+L_{x'|x} e^{-\beta \hbar \omega_{xx'}} +\sum_{y \neq x',x} L_{x'|y} e^{-\beta \hbar \omega_{yx'}}=0.
\end{equation*}
Since $L_{x'|x}\geq 0$ for $x'\neq x$ and $L_{x'|x'}\leq 0$, we thus arrive at
\begin{equation}
L_{x'|x}\leq|L_{x'|x'}| e^{-\beta \hbar \omega_{x'x}}.
\end{equation}
Using the above and the arithmetic-geometric mean inequality,
\begin{equation}
\sqrt{L_{x'|x}L_{y'|y}}\leq\frac{L_{x'|x}+L_{y'|y}}{2},
\end{equation}
and recalling the definition of $t^{x'|x}_{y'|y}$ and $\gamma_{x'y'}$ we arrive at Eq.~\eqref{eq:transport_bound}.

\end{document}